\documentclass[pre,%
reprint,
superscriptaddress,
nofootinbib,
amsmath,amssymb,
aps,
longbibliography
]{revtex4-1}
\usepackage{graphicx}
\usepackage{dcolumn}
\usepackage{bm}
\usepackage{caption}
\usepackage{float}
\usepackage{lipsum}
\usepackage{amsmath}
\usepackage{mathtools}
\usepackage{amssymb}
\usepackage{dsfont}
\usepackage{color}
\usepackage{tikz}
\usetikzlibrary{automata, positioning, arrows}
\usepackage{nccmath}

\usepackage[utf8]{inputenc}
\usepackage{babel}
\usepackage{amsmath,amsthm,amssymb}
\usepackage{graphicx}
\usepackage{subfig}
\usepackage{physics}
\usepackage{listings}
\usepackage{enumitem}
\usepackage{microtype}
\usepackage[raiselinks=true, linktoc=all,hidelinks]{hyperref}
\usepackage[capitalise]{cleveref}
\newcolumntype{L}{>{$}l<{$}}
\newcommand*\diff{\mathop{}\!\mathrm{d}}

\newcommand*\ham{\mathcal{ H }}

\newcommand*\const{\text{const}}

\newtheorem{theorem}{Theorem}
\allowdisplaybreaks[1]

\usepackage[,
text={7.3in,10in},
]{geometry}

\begin{document}


\title{Dynamical Quantum Phase Transition Without An Order Parameter}
\begin{abstract}
    Short-time dynamics of many-body systems may exhibit non-analytical behavior of the systems' properties at particular times, thus dubbed dynamical quantum phase transition. Simulations showed that in the presence of disorder new critical times appear in the quench evolution of the Ising model.  We study the physics behind these new critical times. We discuss the spectral features of the Ising model responsible for the disorder-induced phase transitions. We found the critical value of the disorder sufficient to induce the dynamical phase transition as a function of the number of spins. Most importantly, we argue that this dynamical phase transition while non-topological lacks a local order parameter.
    
\end{abstract}
\author{O. N. Kuliashov}
\email{kulyashov.on@phystech.edu}
\affiliation{Russian Quantum Center, Moscow 121205, Russia}
\affiliation{Moscow Institute of Physics and Technology, Department of General and Applied Physics, Dolgoprudny, 141701, Russia}
\author{A. A. Markov}
\affiliation{Russian Quantum Center, Moscow 121205, Russia}
\affiliation{Lomonosov Moscow State University, Faculty of Physics, Moscow 119991, Russia}
\author{A. N. Rubtsov}
\affiliation{Russian Quantum Center, Moscow 121205, Russia}
\affiliation{Lomonosov Moscow State University, Faculty of Physics, Moscow 119991, Russia}

\maketitle
\section{Introduction}

Despite huge progress in recent years, out-of-equilibrium collective phenomena are far worse understood than equilibrium phenomena \cite{eisert2015quantum}. Even the definitions of a `phase' and `phase transition' are not yet quite clear. Dynamical Quantum Phase Transitions (DQPT) \cite{heyl2013dynamical,heyl2018dynamical} are one of the more established and elaborated attempts to build such an understanding. Equilibrium phase transitions occur when a system substantially and sharply changes its properties as some parameter is varied. For example, it could be temperature or concentration. The dynamical quantum phase transition is a sharp change of systems' properties, happening as the time progresses. Let us consider the analogy more closely.

Equilibrium phase transitions are accompanied by singularities in thermodynamic potentials \cite{landau2013statistical}. Suppose, a varied parameter is temperature and we are trying to find a critical point $T_c = 1/\beta_c$. In canonical ensemble we should look at the points of singularity of the free energy per particle $f(\beta)$ as a function of the inverse temperature $\beta = 1/T $: 

\begin{equation}
\begin{split}
\mathcal{Z}(\beta) &= \Tr e^{-\beta \mathcal{H}} = \sum_i \langle{\psi_i}\vert e^{-\beta \mathcal{H}}\vert{\psi_i}\rangle\\
f(\beta) &= - N^{-1} \ln \mathcal{Z}(\beta).
\end{split}
\end{equation}
Here $\vert{\psi_i}\rangle$ stands for any basis, \(N\) is the number of particles in the system and $\mathcal{Z}(\beta)$ is partition function. In finite systems partition function is a finite sum of exponents, therefore it is entire as a function of the complex temperature $z$. Therefore, derivatives of the free energy $f(z)$
\begin{equation}
\label{eq:f_derivative}
\frac{d f(z)}{dz} =  -N^{-1} \frac{1}{Z}\frac{d Z(z)}{dz}, 
\end{equation}
can diverge only at the zeros of the partition function $\mathcal{Z}(z)$. In a finite system, $z$ coordinate of a partition function zero, called Fisher zero, must have an imaginary part. Otherwise the partition function is a sum of positive numbers and can not be zero. However, a Fisher zero can approach the real inverse-temperature line in the thermodynamic limit \cite{yang1952statistical, fisher1965nature,bena2005statistical}. Thereby there appears a singularity in the free-energy and a phase transition at a real temperature $T_c = 1/\beta_c$ in the thermodynamic limit.  

 The theory of DQPTs  is built in close formal analogy to the theory of equilibrium phase transitions \cite{heyl2013dynamical}. Let us for simplicity consider the case of quench dynamics. Suppose that a system is prepared in the ground state $\ket{\psi_0}$ of the Hamiltonian $\mathcal{H}_0$. Then the system is evolved by a different Hamiltonian $\mathcal{H}$. The key observation, leading to the concept of DQPT, is that the Loschmidt amplitude \(\mathcal{G}(t)\)\footnote{Note that the definition of Loschmidt amplitude is not consistent in the literature. We follow here that of the Ref. \cite{heyl2018dynamical}. While it is used frequently in literature on DQPT, in more general context a different convention $\mathcal{G}'(t) = \bra{\psi_0}e^{-it \ham_0} e^{-it \ham}\ket{\psi_0}$ is more typical. The definitions agree up to a phase factor $e^{iE_0t}$.} is formally similar to a partition function of an equilibrium system at an imaginary temperature:

\begin{equation}
    \label{eq:loschm_def}
    \mathcal{G}(t) = \bra{\psi_0} e^{-it \ham}\ket{\psi_0} \longleftrightarrow \mathcal{Z}(\beta) = \sum_i \bra{\psi_i} e^{-\beta \mathcal{H}}\ket{\psi_i}.
\end{equation}
It is more convenient to work with the probability $\mathcal{L}(t) = |\mathcal{G}(t)|^2$ called Loschmidt Echo (LE) rather than with the amplitude $\mathcal{G}(t)$.

As $\mathcal{G}(t)$ is interpreted as a dynamical partition function, the rate function \(\lambda(t) = - N^{-1} \log \mathcal{L}(t) \) can be considered analogous to the free-energy per particle,  with the time being interpreted as a complex inverse temperature:
\begin{equation}\label{eq:analogy}
    \lambda(t) = - N^{-1} \ln \mathcal{L}(t) \leftrightarrow f(\beta) = - N^{-1} \ln \mathcal{Z}(\beta).
\end{equation}

Equilibrium phases are separated by the points in parameter space, where free energy per particle $f(\beta)$ has a singularity. The correspondence in \cref{eq:analogy} suggests to inspect closely the points where the rate function is non-analytic, or equivalently \(z_{LE}\) - the zeros of the LE.  The Loschmidt echo is a measure of probability to find a system in the state it was prepared in. When LE is equal to zero, an instantaneous sate is orthogonal to the initial state. So, intuitively, the critical times, when LE is zero, correspond to substantial changes in the system's state.



It was demonstrated \cite{heyl2013dynamical} that these critical times might indeed correspond to an interesting dynamical process, dubbed the dynamical quantum phase transition. Later works have shown that the similarities between equilibrium and dynamical quantum phase transitions can be pushed much further the formal analogy. In many cases  (with some exceptions \cite{vajna2014disentangling}) DQPT occurs during a quench across an underlying equilibrium phase transition \cite{heyl2018dynamical,heyl2013dynamical,karrasch2013dynamical,schmitt2015dynamical,schmitt2015dynamical,vajna2015topological}. That is, when initial and post-quench Hamiltonians correspond to different equilibrium phases. An analog of first-order phase transitions was suggested in Ref.\cite{canovi2014first}. Topological phase transitions \cite{thouless1982quantized} have a non-equilibrium counterpart as well \cite{vajna2015topological,budich2016dynamical, schmitt2015dynamical}. At least some of the DQPT obey the dynamical scaling defined by a corresponding out-of-equilibrium analog of the universality class \cite{heyl2015scaling,trapin2021unconventional}. 

Having these similarities to the equilibrium phase transition, a natural question is whether any kind of order parameter exists that can signal DQPT. In the case of the first observed DQPT in the transverse-field Ising model \cite{jurcevic2017direct}, the answer is positive. In Ref.~\cite{heyl2013dynamical} the Ising chain was quenched through an underlying phase transition between ferromagnetic and paramagnetic phases. In this case, the longitudinal magnetization oscillates precisely with the period corresponding to the critical time. More generally, a similar conclusion can be drawn for systems that undergo a DQPT across a symmetry-braking phase transition \cite{heyl2018dynamical,weidinger2017dynamical}. Another interesting idea in this direction is to introduce a localized version of the free energy \cite{halimeh2021local}. Topological dynamical quantum phase transitions were shown to have a non-local order parameter \cite{budich2016dynamical}. Whether a local order parameter exists for non-topological phase transitions have been an open question \cite{heyl2018dynamical}.

In the present manuscript we study a disorder-induced dynamical quantum phase transition in the Transverse Field Ising Model (TFIM), first numerically observed in Ref.~\cite{cao2020influence} and possible local order parameters for the transition. We argue that this phase transition is local in \(k\)-space and can be attributed to a singularity in a Bardeen-Cooper-Schrieffer (BCS) wave-function of the corresponding fermionic model. We find a lower bound for the disorder amplitude required to cause the transition as a function of the number of spins in the system. Furthermore, we demonstrate that a large class of local order parameters can not be used to witness the phase transition. As we shall see, the phase transition is also not a topological one, therefore the phase transition does not fit into the usual equilibrium categories.  

The article is organized as follows: in \cref{sec:model} we describe the model we are working with and techniques for the calculation of LE and spin-spin correlators in TFIM. In \cref{sec:main_res} we summarize our main results: the appearance of the second series of DQPTs and the lack of their influence on the observables in the system. In \cref{sec:theory_zeros} we derive disorder-induced corrections for Fisher zeros and obtain conditions on the disorder necessary for the emergence of the new DQPTs. In \cref{sec:theory_correls} we derive theoretical bounds on the influence of disorder on observables and show that the dynamics of correlators near the time-critical point may be influenced by Fisher zeros arbitrarily far from the critical point. Finally, in \cref{sec:conclusion} we discuss how our findings might influence the general framework of DQPT.

\section{System and Method} \label{sec:model}

We look at the following transverse field Ising model with periodic boundary conditions:
\begin{equation} \label{eq:main_ham}
    \ham(\{h_i\}) = -J\sum_{i=1}^N \sigma_i^x \sigma_{i+1}^x +\sum_{i=1}^N h_i \sigma_i^z.
\end{equation}
Here \(N \gg 1\) is the number of spins, $i$ denotes the position of a spin and the coupling constant is set to \(J = 1\) from now on. Initially, the system is prepared in the ground state of the Hamiltonian \cref{eq:main_ham} with zero magnetic field on all the sites \(\forall i \; h_i^0 = 0\). Thus, the system is prepared in the ferromagnetic phase. At time \(t = 0\) the on-site magnetic fields are suddenly changed:

\begin{equation}\label{eq:quench}
    \ham_0 (\{h_i^0\}) \longrightarrow \ham_1 (\{h_i^1\})
\end{equation}

 The after-quench Hamiltonian has random fields distributed around a value exceeding the critical field in the Ising model:   
 
\begin{equation}\label{eq:dis_distr_def}
    \forall i \; h^1_i = h^1 + \delta_i \, , \, \delta_i \in \mathcal{U}_{[-D, D]} , \; h^1 > h_{crit} =1,
\end{equation}
where \(\mathcal{U}_{[-D, D]}\) describes a uniform random distribution. 

For such a quench two series of critical times are observed \cite{cao2020influence}. The first is the prototypical DQPT timescale connected to the ferromagnetic-paramagnetic phase transition. Secondly, there is a new series of critical times induced by disorder \cite{cao2020influence}. The new DQPT is in the focus of the present study. 

For the Ising model one can exactly calculate all the necessary quantities: the Loschmidt echo, the position of the Fisher zeros, and the average values of observables. This is possible due to the mapping \cite{jordan1993paulische} of the spins to free fermions with the Hamiltonian:
 
 \begin{equation}
        \ham(\{h_i\}) =-\sum_{i=1}^{N}c_{i}^{\dagger} c_{i+1}+c_{i}^{\dagger} c_{i+1}^{\dagger}+h.c.
        - 2 h_i c_{i}^{\dagger} c_{i}
         \label{eq:main_ham_fermions_pos_bas} 
\end{equation}

Hamiltonian of this form can be diagonalized in terms of Bogoliubov quasi-particles $\eta_\alpha$ \cite{bogoliubov1947theory,lieb1961two}, which are related to the $c$-operators by a unitary transformation, mixing creation and annihilation operators. For both the initial $\ham_0$ and the post-quench Hamiltonian $\ham_1$ we can write:   
\begin{equation} 
\begin{split}
        &\ham_0 = \sum_{\alpha=1}^{N} E^0_\alpha {\eta^0_\alpha}^{\dagger}\eta^0_\alpha
    \label{eq:main_ham_fermions_e_bas}\\
         &\begin{pmatrix}
    \eta^0 \\
    {\eta^0}^\dagger
    \end{pmatrix} = 
    U_0 \begin{pmatrix}
    c \\
    c^\dagger
    \end{pmatrix}\\
    \end{split} \qquad \qquad
    \begin{split}
         &\ham_1 =\sum_{i=\alpha}^{N} E_\alpha \eta_\alpha^{\dagger}\eta_\alpha\\
    &    \begin{pmatrix}
    \eta \\
     \eta^\dagger
    \end{pmatrix} = 
    U \begin{pmatrix}
    c \\
    c^\dagger
    \end{pmatrix}.
\end{split}   
\end{equation}
Here, the operators without a varying index denote the sets of creation and annihilation operators. For example in case of \(c\)-operators the notation should read:
 \(c^{\dagger}\equiv \{c^{\dagger}_1,...c^{\dagger}_N\}\) and \(c \equiv \{c_1,...c_N\}\). Thereafter, the same convention for operators without a varying index is used. The matrices \(U\) and \(U_0\) are assumed to be unitary to keep the canonical commutation relations among the operators \(\eta\). The energies \(E_\alpha\) and \(E^0_\alpha\) are chosen to be positive. Thus, the ground states of the Hamiltonians are the vacuum states \(\ket{vac}_{\eta^0}\)  and \(\ket{vac}_{\eta}\)  annihilated by all the \(\eta_\alpha\) and \(\eta^0_\alpha\) operators correspondingly. 


The initial state is the ground state of the Hamiltonian \(\ham_0\) with all the \(h_i\) set to zero. We will work in the sector with the even number of fermions. Such fermionic initial state corresponds to the fully polarized 'Schrodinger cat' spin state \cite{lieb1961two}:  

\begin{equation}\label{eq:init_state}
    \ket{\psi(0)} = \frac{\ket{\rightarrow} + \ket{\leftarrow}}{\sqrt{2}},
\end{equation}
where we denote by \(\ket{\leftarrow}\) and \(\ket{\rightarrow}\) two degenerate lowest energy eigenstates in which all spins are polarized along the \(x\) axis to the left and to the right respectively.

Now let us consider the evolution of the state \cref{eq:init_state}. The operators $\eta_i$ have a very simple dynamics \(\eta_\alpha(t) = \exp(-iE_\alpha t)\eta_\alpha(t)\).  Therefore, we can readily obtain the instantaneous state, once the initial state $\ket{vac}_{\eta^0}$ is expressed in terms of operators $\eta_\alpha$. This is possible due to an extension of the Thouless theorem \cite{thouless1960stability}, see Appendix E3 of Ref. \cite{ring2004nuclear}. It tells that there exists an antisymmetric matrix  $G_{\alpha\beta}$, such that the initial state $\ket{vac}_{\eta^0}$ can be related to the vacuum of $\ket{vac}_{\eta}$ as follows:

\begin{equation}\label{eq:thouless}
   \ket{vac}_{\eta^0} =  \frac{1}{\mathcal{N}}\exp\left(\sum_{\alpha\beta}\eta_\alpha^\dagger G_{\alpha\beta} \eta_\beta^\dagger\right)\ket{vac}_{\eta},
\end{equation}

where \(\mathcal{N}\) is a normalization coefficient. Thus, the post-quench state assumes the following form:  

\begin{equation}\label{eq:psi_t_full}
    \ket{\psi(t)} = \frac{1}{\mathcal{N}}\exp\left(\sum_{\alpha\beta}\eta_\alpha^\dagger(t) G_{\alpha\beta} \eta_\beta^\dagger(t)\right)\ket{vac}_{\eta}.
\end{equation}

The state $\ket{\psi(t)}$ has a form of the famous Bardeen-Cooper-Schrieffer wave function, which is a coherent bosonic state, with bosons formed by pairs of fermions. In \cref{eq:psi_t_full} the \((\alpha, \beta)\) element of the matrix \(G\) indicates the presence of a Cooper pair formed by \(\eta_\alpha\) and \(\eta_\beta\) modes. We will call the matrix \(G\) the BCS matrix, therefore. The matrix \(G\) can be found from the condition that \cref{eq:psi_t_full} is the ground state of the post-quench Hamiltonian (for more details see \cite{zhong2011loschmidt}). 

The Loschmidt echo and Fisher zeros are fully determined by the matrix \(G\) and energies \(E_\alpha\) of \(\ham_1\) \cite{zhong2011loschmidt}: 
 
 \begin{equation}\label{eq:loschm_calc}
        \mathcal{L}(t) = \prod_{\alpha, \beta>\alpha}\left(1 - \frac{4G_{\alpha\beta}^2}{(1 + G_{\alpha\beta}^2)^2}\sin^2\left(\frac{E_\alpha + E_\beta}{2} t\right) \right).
    \end{equation}
    
We will look at the zeros of the boundary partition function 
\begin{equation}\label{eq:bpartition}
    Z(z) = - N^{-1} \log \bra{\psi_0} e^{-z\mathcal{H}} \ket{\psi_0},
\end{equation}
which are called Fisher zeros \cite{fisher1965nature} and are connected with zeros of LE via \(z = i z_{LE}\). Thus, purely imaginary Fisher zeros correspond to real critical times. By requiring that the LE is equal to zero, \(\mathcal{L}(i z) = 0\), we obtain the coordinates of the Fisher zeros:

\begin{equation}\label{eq:zeros_calc}
    z_n (\alpha, \beta) = \frac{1}{E_\alpha + E_\beta} \left(\ln\abs{G_{\alpha\beta}}^2 + i(2n+1)\pi\right)
\end{equation}

From \cref{eq:zeros_calc} it is clear that the calculation of all Fisher zeros has the same computational complexity as the calculation of all elements of matrix \(G\). In turn, matrix \(G\) is expressed as \(G = -W_1^{-1}W_2\) where 
\begin{align}
U_0 U^{-1} = \begin{pmatrix}
W_1 & W_2 \\
W_2^* & W_1^*
\end{pmatrix}
\end{align}
and matrices \(U, U_0\) are from \cref{eq:main_ham_fermions_e_bas}. In other words, we need to diagonalize initial and final Hamiltonians -  matrices of size \(2N \times 2N\), and then multiply matrices of sizes \(2N \times 2N\) and \(N \times N\). Thus, the calculation of the Fisher zeros has the same computational complexity as matrix multiplication. 

The time evolution of the average values of a product of spin operators is discussed in detail in Ref.~\cite{barouch1971statistical2}. Let us outline the scheme. First, one expresses the spin-spin correlators between sites $n$ and $m$  in terms of chains of fermionic operators:
\begin{align}\label{eq:spincorr_th_ferm}
\begin{aligned}
    &\langle\sigma^x_m \sigma^x_n\rangle = \bra{\psi_0(t)} \sigma^x_m \sigma^x_n \ket{\psi_0(t)} = \\ 
    &\bra{\psi_0(t)} (c_m^\dagger + c_m) \exp(\pi i \sum_{m}^{n-1} c_i^\dagger c_i)(c_n^\dagger + c_n)\ket{\psi_0(t)} 
\end{aligned}
\end{align}
Second, using Wick's theorem we reduce the problem to the calculation of pfaffians of matrices constructed from pair-wise fermionic correlators:
\begin{equation}\label{eq:corrx_pfaff}
    \langle \sigma_m^x \sigma_n^x \rangle = \text{pf} \begin{pmatrix}
    \left(S^{mn}\right) & \left(G^{mn}\right) \\
    -\left(G^{mn}\right)^\dagger & \left(Q^{mn}\right)
    \end{pmatrix},
\end{equation}
with
\begin{align}
    \left(S^{mn}\right)_{ij} = \langle (c_i^\dagger - c_i)(c_j^\dagger - c_j) \rangle + \delta_{ij}, \; &l \leq i, j \leq m-1 \nonumber \\
    \left(Q^{mn}\right)_{ij} = \langle (c_i^\dagger + c_i)(c_j^\dagger + c_j) \rangle - \delta_{ij}, \; &l+1 \leq i, j \leq m \nonumber \\
    \left(G^{mn}\right)_{ij} = \langle (c_i^\dagger - c_i)(c_j^\dagger + c_j) \rangle, \phantom{-\delta_{ij}} \; &l \leq i \leq m-1 , \nonumber\\ &l+1 \leq j \leq m.\nonumber
\end{align}
For correlators along the \(x\)-axis, the size of the matrices \(S, Q, G\) is \(\abs{m - n}\). For the z-component correlators \(\langle \sigma_m^z \sigma_n^z \rangle\)  the formula is similar to \cref{eq:corrx_pfaff}, but the \(S, Q, G\) matrices' size is always ~\(2\) independently of \(m\) and \(n\).

Thus, for spins at a distance \(\abs{m - n} = d\) from each other, the calculation of \(x\)-oriented correlators is reduced to the calculation of the pfaffian of a \(2d \times 2d\) matrix. Such an operation has time asymptotics \(\mathcal{O}(d^3)\). For arbitrarily separated spins this comes to \(\mathcal{O}(N^3)\). Calculation of \(z\)-oriented correlators always requires calculation of pfaffians of \(4 \times 4\) matrices. Hence, it has only \(\mathcal{O}(1)\) asymptotics. 

To calculate these correlators at an arbitrary time we need to find \(c_i(t)\). To do so, we first we find the evolution of the eigenmodes \(\eta_\alpha (t) = \eta_\alpha (0) e^{-i t E_\alpha}\) and use \cref{eq:main_ham_fermions_e_bas} to find \(c_i(t)\).

\section{Main Results} \label{sec:main_res}
\begin{figure*}[!htb]
  \centering
    \includegraphics[width=0.8\textwidth]{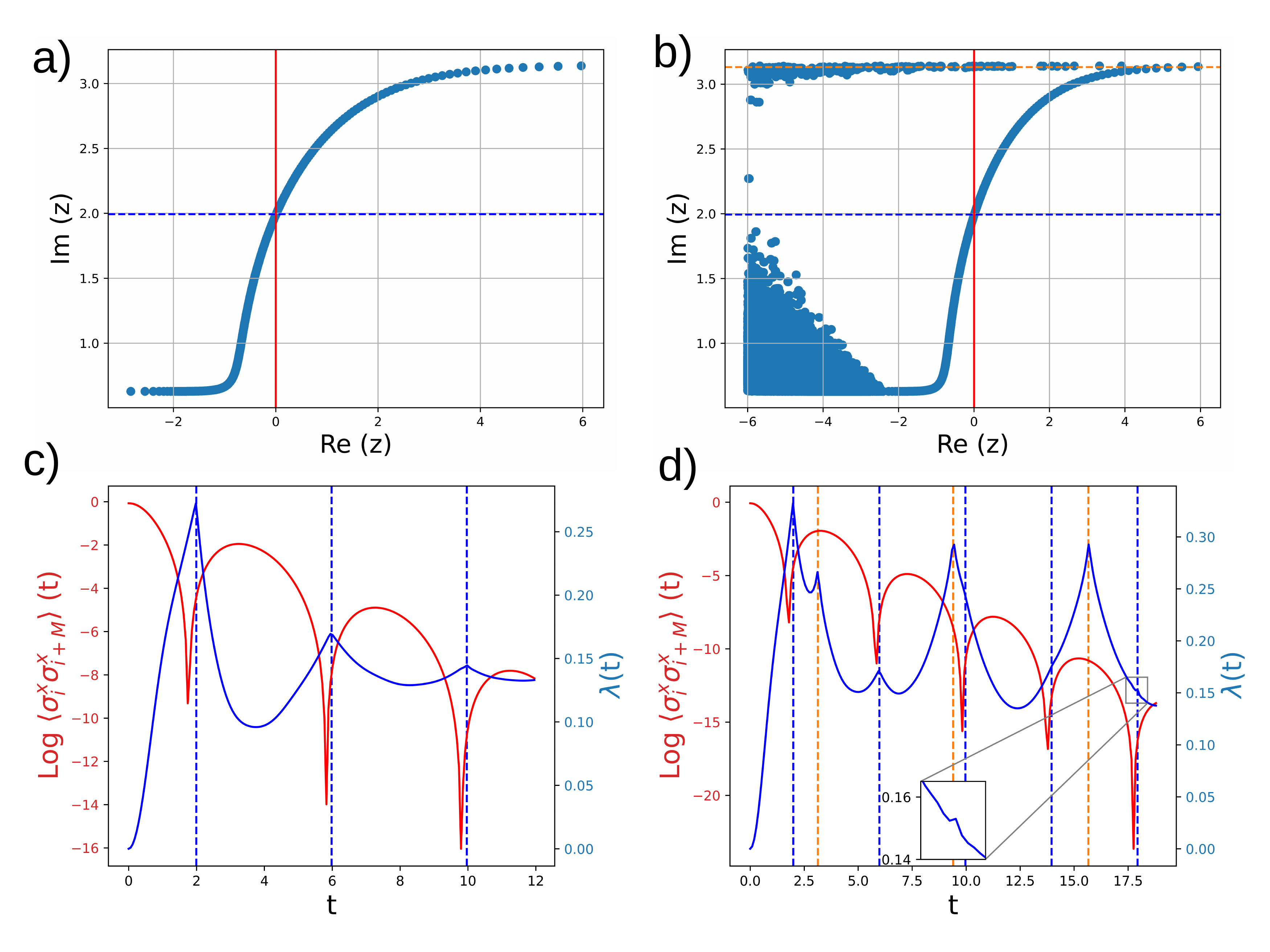}
    \caption{Fisher zeros (top row), Loschmidt rate function and spin-spin correlators (bottom row) for a quench from ferromagnetic phase (\(h^0 = 0.5\)) to paramagnetic phase (with the mean value of magnetic fields \(h^1 = 1.5\)). Chain length is \(N = 1000\). Plots on the left are calculated for a homogeneous external field \(h^1\), plots on the right - for a disordered \(h^1\) with the disorder strength \(D = 0.001\) - see \cref{eq:dis_distr_def}. Plots \textbf{(a), (b)} show Fisher zeros, calculated from \cref{eq:zeros_calc}. A DQPT is induced by a zero lying on the imaginary axis (red line). The period of such a DQPT is twice the imaginary component of said zero. With a homogeneous external magnetic field there is only one such zero (plot \textbf{(a)}), with a blue dashed line drawn on the level of the imaginary component of this zero \(t_1 = 2.0\). With a disordered external field there are two such zeros (plot \textbf{(b)}): one on the same level \(t_1 = 2.0\) also marked with a blue dashed line, and one at \(t_2 = 3.1\) marked with an orange dashed line. The plots \textbf{(c), (d)} show the logarithm of spin-spin correlators (spin-spin distance \(d = N / 2 = 500\)) in red and the Loschmidt rate function in blue. With a homogeneous external field (plot \textbf{(c)}), both spin-spin correlators and the Loschmidt rate function oscillate with a period \(2 \cdot t_1 = 4.0\) corresponding to the purely imaginary zero \(i \cdot t_1\) (plot \textbf{(a)}). With a disordered external field (plot \textbf{(d)}), Loschmidt rate function (blue) has non-analyticities with both periods \(2 \cdot t_1 = 4.0\) and \(2 \cdot t_2 = 6.2\) corresponding to the two imaginary zeros in the plot \textbf{(b)}. Spin-spin correlators still have non-analyticities with only one period \(2 \cdot t_1 = 4.0\).} 
    \label{fig:fisher_zeros_no_small}
\end{figure*}
In the weak disorder limit, the rate function develops periodically appearing kinks corresponding to two series of the DQPT as shown in \cref{fig:fisher_zeros_no_small}a. The first, \(t_1 = (\frac{1}{2} + n)t_1^*\) can be attributed to the ferromagnetic-paramagnetic equilibrium phase transition, as highlighted by the behavior of spin-spin correlators approaching zero at these times. The second, \(t_2 = (\frac{1}{2} + n)t_2^*\), corresponding through the \cref{eq:zeros_calc} to the lowest energies \(E_{min}^{-1}\) of \(\ham_1\) is more mysterious. 

These DQPTs do not correspond to an equilibrium phase transition in the disordered Ising chain \cite{sachdev1999quantum}. As we can see in ~\cref{fig:fisher_zeros_no_small} it does not alter the spin-spin correlators, reflecting that the transition is not connected to the order-disorder phase transition. As we shall see, it is difficult for almost all observables to trace this new phase transition. 

Another important insight from the behavior of spin-spin correlators is that the disorder-induced phase transition is not a topological one. The spin-spin correlators are mapped to the string order parameter for the corresponding Kitaev chain \cite{kitaev2001unpaired}. As we can see, the topological order parameter is not affected by the presence of the disorder.

To explain these phenomena, we analytically obtained the corrections to the matrix \(G_{\alpha\beta}\) from \cref{eq:psi_t_full}, and therefore for the post quench wave-function:

\begin{align}\label{eq:psi_d}
\begin{aligned}
    &\ket{\psi(t)}_D = \\  
    &\frac{1}{\mathcal{N}}\exp\left(\sum_{\alpha\beta}\eta_\alpha^\dagger(t) \left(G_{\alpha\beta} + \Delta G_{\alpha\beta}\right)\eta_\beta^\dagger(t)\right)\ket{vac}_{\eta}
\end{aligned}
\end{align}

Only the lowest energy part of the matrix \(G\) is extremely sensitive to disorder, with the rest of the terms being insensitive:

\begin{align}\label{eq:g_sensetiv}
    D = \Theta(N^{-3}) \text{ leads to } \Delta G_{\alpha\beta} = \Theta(1) \nonumber \\ \text{ at } \abs{E_\alpha - E_{min}}, \abs{E_\beta - E_{min}} = \mathcal{O}(N^{-2})
\end{align}

Physically, \cref{eq:g_sensetiv} might be interpreted as a substantial change in the BCS wave function for pairs formed by low-energy excitations when a very weak disorder is introduced. This sensitivity leads to the appearance of the \(t_2\) series of the DQPT at the disorder strength \(D = \Theta(N^{-3})\). For disorder strength \(D = o(N^{-3})\) there is no second series of DQPTs. Thus, we shall concentrate our attention on the threshold regime \(D = \Theta(N^{-3})\) in most cases. Numerical results suggest that our analytical results are still valid for larger disorder amplitudes see \cref{sec:app:finitedisorder}.

Next, we shall prove the lower bound for the change in fermionic correlators of the form:
\begin{equation}\label{eq:ferm_scale_intro}
    \Delta \langle c_i c_{i+n} \rangle = \mathcal{O}\left(N^{-1}\right).
\end{equation}
In \cref{sec:theory_correls}, we shall see that this means that the short-range spin-spin correlators \(\abs{i-j} \ll N\) follow the bound
\begin{equation} \label{eq:x_correl_change}
    \Delta \langle \sigma_i^x \sigma_j^x \rangle = \mathcal{O}(N^{-1}).
\end{equation}
And more generally, for \(n\)-spin correlators:
\begin{equation}\label{eq:ferm_scale_intro_many}
    \Delta \underbrace{\langle \sigma_{i1}^x \sigma_{i2}^x \dots \sigma_{in}^x \rangle}_{n \text{ spins}} = \mathcal{O}(N^{-1}),
\end{equation}
but only as long as the maximal distance \(d_{max}\) between any two spins in the correlator is independent of \(N\) and \(d_{max} \ll N\). Same is true for for \(z\)-oriented correlators. Thus, all local spin correlators change negligibly in the \(N \rightarrow \infty\) limit.

This leaves open the question as to whether the long-range spin-spin correlations might be used to witness the phase transition. In a large yet finite system, the long-range correlations can build up in finite time \(N/v_{LR}\), where \(v_{LR}\) is the Lieb-Robinson velocity \cite{lieb1972finite} of the system. Numerically we observe that the long-range spin-spin correlations are also insensitive to the disorder-induced DQPTs, see \cref{fig:fisher_zeros_no_small}. Although we do not give a rigorous proof, an intuitive argument can be given.

\begin{figure}[h!]
    \centering
    \includegraphics[width=0.4\textwidth]{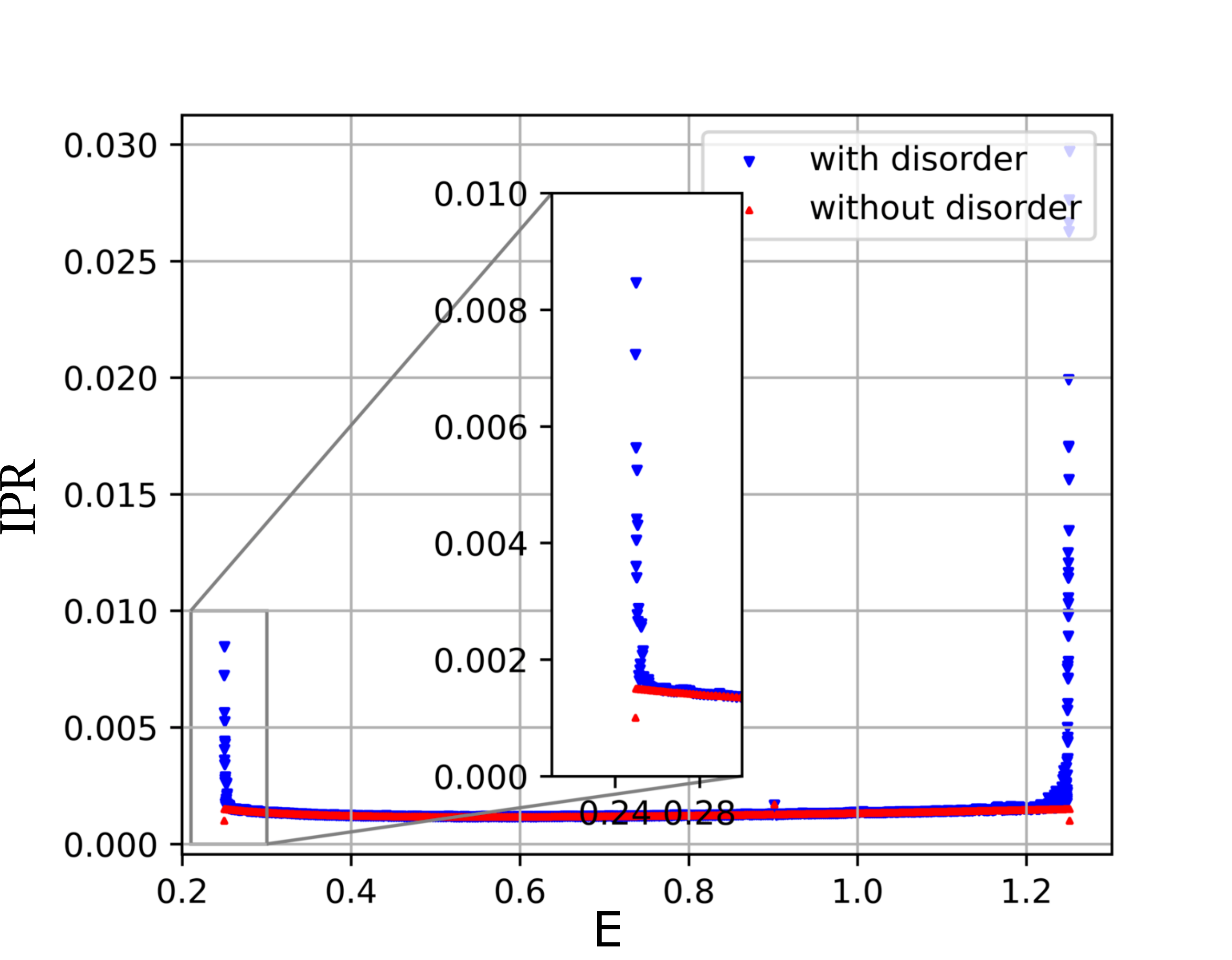}
    \caption{The plot shows the Inverse Participation Ratio (IPR) of eigenstates with a given energy. IPR is a measure of localization.  For an eigenstate \(\ket{\psi_\alpha} = \eta_\alpha^\dagger \ket{vac}_{\eta}\) with energy \(E_\alpha\) the IPR is calculated as \(IPR (E_\alpha) = \sum\limits_i \bra{\psi_\alpha} c_i^\dagger c_i\ket{\psi_\alpha}^2 = \sum\limits_i p_i^2\) where \(p_i\) is the probability for the particle to be on the \(n\)-th site. IPR is 1 for a fully localized state and \(1/N\) for a fully delocalized one.  Red dots show IPR for eigenstates in a homogeneous system. We can see, that all eigenstates are fully delocalized (\(IPR = 1 / 1000\)). Blue dots show IPR for eigenstates in a disordered system. The inset shows that in a disordered system low-energy states have higher IPR. Therefore, these states are localized.}
    \label{fig:ipr_population}
\end{figure}

From \cref{fig:ipr_population} we see that the low-energy excitations are localized, allowing us to rewrite \cref{eq:psi_d}:
\begin{align}
\begin{aligned}
    \ket{\psi(t)}_D &= \exp\left(\sum_{\alpha\beta}\eta_\alpha^\dagger(t) \left(\Delta G_{\alpha\beta}\right)\eta_\beta^\dagger(t)\right)\ket{\psi(t)} \\
    &= \sum_i R_i \ket{\psi(t)},
    \end{aligned}
\end{align}
where \(\ket{\psi(t)}\) is the post-quench wave function with \(D = 0\) (from \cref{eq:psi_t_full}), and \(R_i\) are local operators corresponding to the low-energy excitations. Thus one can expect that weak disorder does not affect the long-range correlations.
 
From a computational perspective, our findings allow for faster analysis of the LE and correlators in the presence of non-analyticities. Straightforwardly, we can see if the LE has non-analyticities by calculating all the Fisher zeros and checking if any of them approach the imaginary axis. But as was discussed beneath \cref{eq:zeros_calc}, this will require us to multiply matrices of size \(N \times N\), which takes \(\mathcal{O}(N^p) \, , \, 2 < p < 3\) operations. Furthermore, to check if logarithms of spin-spin correlators experience non-analyticities simply by calculating them requires \(\mathcal{O}(N^3)\) operations -- see \cref{eq:corrx_pfaff} and discussion below. 
In contrast, \cref{eq:gij_final_momentum} allows us to judge (for large system size N and small disorders) whether disorder creates additional non-analyticities in the Loschmidt echo and correlators just by looking at the lowest Fourier-components of the magnetic field \(\{h_i\}_{i=1}^N\). Their calculation takes only \(\mathcal{O}(N)\) operations.

\section{Theoretical explanation} \label{sec:theory_zeros}
\subsection{Physical picture}\label{subs:phys_pic}

In \cref{sec:main_res} we explained the observed changes in the Loschmidt echo and correlators by the proliferation of low-energy excitations. Now we qualitatively show why disorder in the transverse field is very effective in producing specifically the low-energy excitations. We shall start by sketching out the physical picture, and then explain the mathematical details in the later sections.

In a homogeneous system the density of excitations in the post-quench state can be expressed in terms of the matrix $G$ in momentum basis and is given by \cite{calabrese2012quantum}:
\begin{equation} \label{eq:gmat_meaning}
    \bra{\psi} \eta^\dagger_k \eta_k \ket{\psi} = \frac{\abs{G_{k,-k}}^2}{1 + \abs{G_{k,-k}}^2}.
\end{equation}
Here index \(k\) denotes quasi-momentum of an excitation. As shown in \cref{fig:gmat_on_k_clean}, \(\abs{G_{k, -k}}^2\) rapidly decreases from \(\Theta(N^2)\) at the lowest energies, to the values close to zero at all the other energies. This means that the low-energy modes \(\langle\eta_k^\dagger \eta_k\rangle \approx 1 \, , \, \abs{k} \approx \pi \, , \, E_k \approx E_{min}\) are most populated by the quench, while the other modes are not \(\langle\eta_k^\dagger \eta_k\rangle \approx 0 \, , \, \abs{k} \not\approx \pi\).

Suppose, we introduce a weak perturbation, which couples \(k\) modes separated by a momentum \(q\). The coupling is suppressed by the inverse energy difference as is usual in first-order perturbation theory. This means that energy levels set far apart (separated by large \(q\)) are coupled more weakly compared to closely lying energy levels.

Roughly speaking, we have only the lowest energy levels filled and they are mostly coupled only to each other by the perturbation. Thus, \(\langle\eta_k^\dagger \eta_{k+q}\rangle\) (and correspondingly \(G_{k, -k + q}\)) change significantly only for \(\abs{k} \sim \pi \, , \, \abs{q} \ll 1\). In the next subsection, we make this argument quantitative and conclude how it affects the positions of new Fisher zeros.

\subsection{General explanation} \label{subsec:gen_expl}
\begin{figure}[t!]
    \centering
    \includegraphics[width=0.4\textwidth]{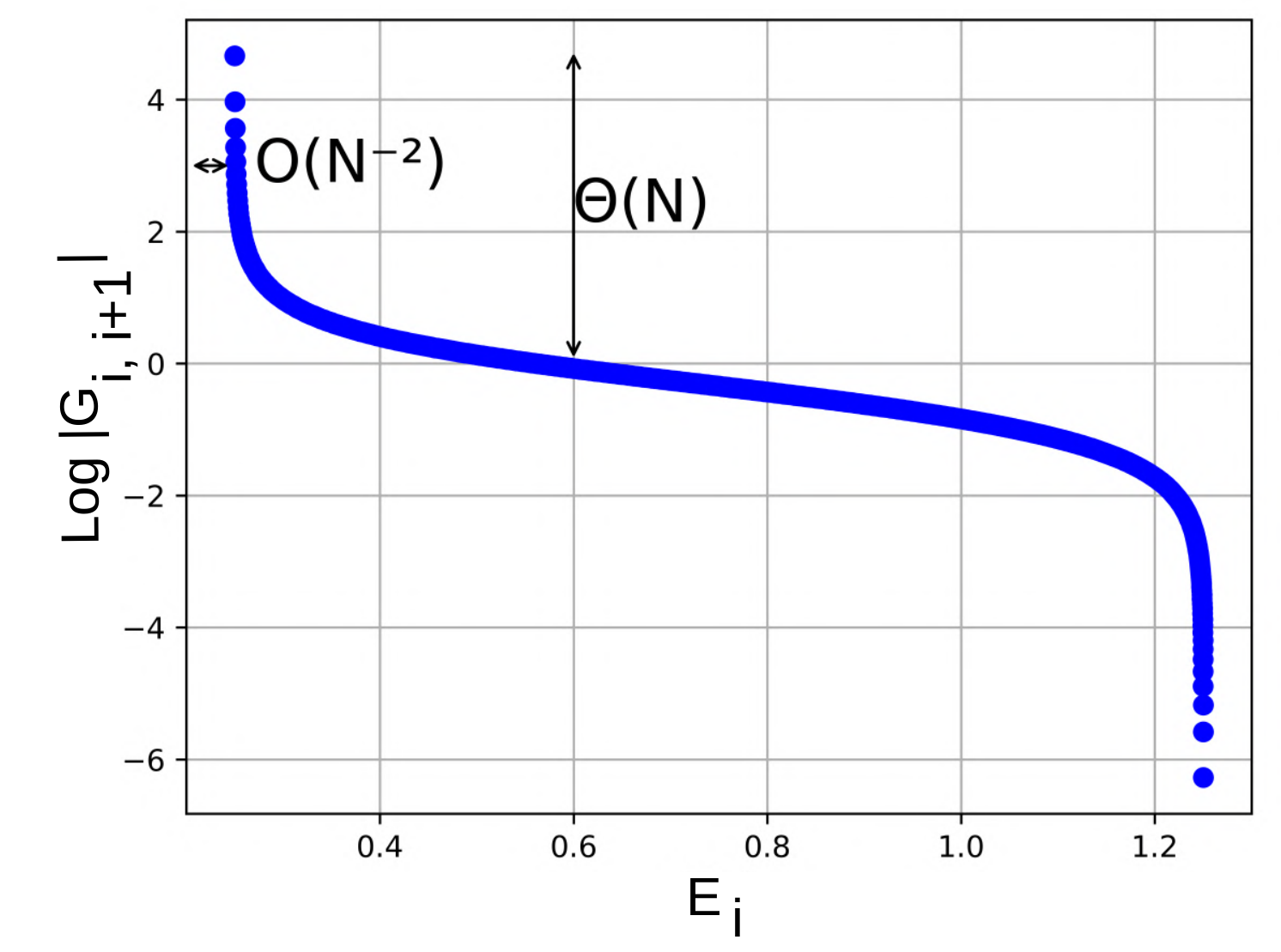}
    \caption{Dependence of the non-zero elements of the BCS matrix \(G\) on the corresponding energy, for a homogeneous external field; logarithmic scale. As explained in \cref{sec:g_shape}, in the homogeneous case the only non-zero elements of matrix \(G\) are \(G_{\alpha, \alpha+1} = - G_{\alpha+1, \alpha}\). These correspond to Cooper pairs of excitations $\eta$ with the same energy and opposite momenta. The dependency rapidly decreases with the energy. The number of pairs surges for the energy \(E_\alpha\) close to the lowest energy in the spectrum \(E_{min}\). For \(\abs{E_\alpha - E_{min}} = \mathcal{O}(N^{-2})\), we get \(G_{\alpha, \alpha+1} = \Theta(N)\) - see \cref{eq:g_asymptot}.}
    \label{fig:gmat_on_k_clean}
\end{figure}

In this section, we focus only on the case of weak disorder in the post-quench Hamiltonian and compare it to the case of homogeneous fields. Weak disorder causes additional non-analytic peaks in the Loschmidt echo. These are the direct consequence of an additional crossing of the imaginary line by Fisher zeros. In \cref{fig:fisher_zeros_no_small} we see that compared to the homogeneous case, the imaginary line is now crossed by a horizontal line of zeroes positioned above all other zeros. Remembering \cref{eq:zeros_calc} for the coordinates of the Fisher zeros,
\begin{equation}
    z_n (\alpha, \beta) = \frac{1}{E_\alpha + E_\beta} \left(\ln\abs{G_{\alpha\beta}}^2 + i(2n+1)\pi\right),
     \tag{\ref{eq:zeros_calc}}
\end{equation}
one can see that since this additional line of zeros is at the very top; it corresponds to the lowest energies. Here and throughout the article we use energy ordered indices, that is \(E_\alpha \leq E_\beta \Leftrightarrow \alpha \geq \beta\). Since this line begins on the very left and intersects the imaginary line, real parts of the corresponding zeros must change from a large negative value (\(-\infty\) in thermodynamic limit) to at least \(0\). Looking at \cref{eq:zeros_calc}, we can see that \(\abs{G_{\alpha\beta}}\) changes from almost zero (\(0\) in thermodynamic limit) to at least \(1\). This observation and those made in \cref{subsec:gen_expl} leads us to the following theorem:
\begin{theorem}\label{theorem:main}
    Suppose the post-quench Hamiltonian contains disorder with the amplitude vanishing as fast as \(D = \Theta(N^{-3})\). Then 
    \begin{enumerate}
        \item The change in spectrum \(E_\alpha\) is vanishingly small in the limit \(N \rightarrow \infty\).
        \item \(G_{\alpha\beta}\) changes in the following way. New entries appear with absolute values up to and inclusive \(\approx 1\); changes in the old entries vanish in the limit \(N \rightarrow \infty\). These new entries correspond to energies \(E_\alpha\) at the bottom of the spectrum, they are in the bottom-right corner of the matrix \(G\)).
    \end{enumerate}
\end{theorem}
If we add weak disorder to the driving Hamiltonian \(\ham_1\) (see \cref{eq:quench,eq:dis_distr_def})
\begin{equation}\label{eq:disord_cond}
    \ham_1 \longrightarrow \Tilde{\ham}_1 = \ham_1 + \mathcal{V} \; ,\; \frac{\norm{\mathcal{V}}}{\norm{\ham_1}} = \nu = \Theta({N^{-3}}) \ll 1,
\end{equation} it leads to corrections to the Fisher zeros' coordinates through the corrections in \(G_{\alpha\beta}\) and \(E_\alpha\). Using perturbation theory we write
\begin{equation}
\begin{aligned} \label{eq:g_e_corrections}
    \Delta{G}_{\alpha\beta} &= \mathcal{D}(G)_{\alpha\beta\gamma\delta} V_{\gamma\delta} + \mathcal{O}(\nu^2),\\
    \Delta{E}_\alpha &= \mathcal{D}(E)_{\alpha\gamma\delta} V_{\gamma\delta}  + \mathcal{O}(\nu^2).
\end{aligned}
\end{equation}
Here, at first glance it may seem that at vanishingly small \(\nu\), corrections to the coordinates of Fisher zeros in \cref{eq:zeros_calc} would also be vanishingly small. That is, even in the presence of disorder there would still be one line of zeros, as in the absence of disorder in \cref{fig:fisher_zeros_no_small} (a). We know this is not the case: see \cref{fig:fisher_zeros_no_small} (b), demonstrating the actual Fisher zeros in the presence of weak disorder. This seeming contradiction is explained by the fact that some of the coefficients \(\mathcal{D}(G)_{\alpha\beta\gamma\delta}\) in \cref{eq:g_e_corrections} diverge with the system size \(N\). 

As we shall see below, the reason for this divergence is the singular behavior of \(G_{\alpha\beta}\), shown in \cref{fig:gmat_on_k_clean}. In the homogeneous case, the only non-zero terms of \(G\) are \(G_{k,-k}\) (here we enumerate \(G\) with quasi-momenta instead of the energy-indices). Moreover, as we demonstrate in \Crefrange{eq:g_sq_exact}{eq:g_approx_deriv}
\begin{align}\label{eq:g_asymptot_k}
 \abs{G_{k,-k}} = \Theta(N)\, , \, \abs{k \pm \pi} = \mathcal{O}(N^{-1}); \nonumber \\ \abs{G_{k,-k}} = \mathcal{O}(1) \text{, other } k.  
\end{align}
 With \(k\) moving away from \(\pi\), \(G_{k,-k}\) decreases as \(\frac{1}{k-\pi}\). In our usual convention with $G$ indexed in the energy basis, this reads:
 \begin{equation}
 \tag{\theequation'}
    \begin{aligned}\label{eq:g_asymptot}
 \abs{G_{\alpha+1,\alpha}} = \abs{G_{\alpha,\alpha+1}} = \Theta(N)\, , \, \abs{E_\alpha - E_{min}} = \mathcal{O}(N^{-2})  \\ \abs{G_{\alpha+1,\alpha}} = \abs{G_{\alpha,\alpha+1}} = \mathcal{O}(1) \text{, for other } E_\alpha.
 \end{aligned}
  \end{equation}
\subsection{Proof of the Theorem 1}

We do not attempt to compute all the coefficients \(\mathcal{D}(G)_{\alpha\beta\gamma\delta}, \mathcal{D}(E)_{\alpha\gamma\delta}\). Instead, we focus on those that are divergent. Moreover, they should diverge fast enough to compensate for \(\nu = \Theta(N^{-3})\). Thus, we will disregard all the coefficients less than \(\Theta(N^3)\), as they will account for vanishing corrections in the \(N \rightarrow \infty\) limit. 

In the following, we obtain explicit expressions for \(\Delta E_\alpha\) (\cref{eq:energy_change}) and for \(\Delta G_{\alpha\beta}\) (\cref{eq:dG_matform}) and show that \(\mathcal{D}(E)_{\alpha\gamma\delta}\) are never divergent, while \(\mathcal{D}(G)_{\alpha\beta\gamma\delta}\) are divergent as \(\Theta(N^{-3})\) for the indices \(\alpha, \beta\) corresponding to the lowest energies.

\begin{proof}
The Bogolyubov operators \(\eta\) corresponding to the Hamiltonian \(\ham_1\) without disorder can be rewritten in terms of \(\Tilde{\eta}\), which are the Bogolyubov operators corresponding to the disordered Hamiltonian, \(\widetilde{\ham} = (\ham_1 + V)\). Thus, we equate the Hamiltonians written in terms of \(\Tilde{\eta}\) and \(\eta\):
\begin{equation}
\begin{aligned}
    \begin{pmatrix}
    \eta \\
    \eta^\dagger
    \end{pmatrix}^\dagger
    \left[
    \frac{1}{2}\begin{pmatrix}
    E & 0 \\
    0 & -E
    \end{pmatrix}
    +
    V
    \right]
    &\begin{pmatrix}
    \eta \\
    \eta^\dagger
    \end{pmatrix}&~=  \\ \nonumber
    \begin{pmatrix}
    \Tilde{\eta} \\
    \Tilde{\eta}^\dagger
    \end{pmatrix}^\dagger
    \frac{1}{2}\begin{pmatrix}
    \widetilde{E} & 0 \\
    0 & -\widetilde{E}
    \end{pmatrix}
    &\begin{pmatrix}
    \Tilde{\eta} \\
    \Tilde{\eta}^\dagger
    \end{pmatrix}.
\end{aligned}
\end{equation}
where \(E, \widetilde{E}\) are diagonal matrices with energies of \(\ham_1, \widetilde{\ham}\) on their diagonals: \(E_{\alpha \alpha} = E_\alpha, \; \widetilde{E}_{\alpha \alpha} = \widetilde{E}_\alpha\). 

Sets of Bogolyubov operators for different Hamiltonians are connected by a unitary transform \(A\):
\begin{equation}
    \begin{pmatrix}\label{eq:A_def}
    \Tilde{\eta} \\
    \Tilde{\eta}^\dagger
    \end{pmatrix} = 
    A \begin{pmatrix}
    \eta \\
    \eta^\dagger
    \end{pmatrix}.
\end{equation}
From this we deduce:
\begin{equation} \label{eq:eta_transform_full}
    \begin{pmatrix}
    E & 0 \\
    0 & -E
    \end{pmatrix}
    +
    2V = A^\dagger \begin{pmatrix}
    \widetilde{E} & 0 \\
    0 & -\widetilde{E}
    \end{pmatrix} A.
\end{equation}

In \cref{eq:eta_transform_full} \(V\) represents a weak perturbation. Thus, \(A = 1 + a\) and \(\Tilde{E} = E + \Delta E\), where \(a\) and \(\Delta E\) have norms of the order of \(V\). Linearizing \cref{eq:eta_transform_full} we obtain:
\begin{equation}
    2V - \begin{pmatrix}
    \Delta E & 0 \\
    0 & -\Delta E
    \end{pmatrix} = a^\dagger  \begin{pmatrix}
    E & 0 \\
    0 & -E
    \end{pmatrix} +  \begin{pmatrix}
    E & 0 \\
    0 & -E
    \end{pmatrix} a.
\end{equation}

In the linear approximation, using the unitarity of \(A\), one can show that \(a^\dagger \approx -a\). Denoting the diagonal components of the disorder potential as \(V_d\) and its non-diagonal components as \(V_{nd}\), we derive:
\begin{equation}\label{eq:energy_change}    
    2 V_d =
\begin{pmatrix}
     \Delta E & 0 \\
    0 & -\Delta E
    \end{pmatrix},    
\end{equation}
which implies that:
\begin{equation}
\mathcal{D}(E)_{\alpha\gamma\delta} = 2\delta_{\alpha\gamma}\delta_{\gamma\delta} \tag{\theequation'}.
\end{equation}
For the off-diagonal terms we obtain:
\begin{equation}\label{eq:a_e_commutator}
    2V_{nd} = \left[ \begin{pmatrix}
    E & 0 \\
    0 & -E
    \end{pmatrix}, a\right].
\end{equation}
We rewrite \(V_{nd}\) and \(a\) in a block-matrix form, using the hermiticity of \(\mathcal{V}\) and the symmetries of \(a\) inherited from \cref{eq:A_def}:
\begin{equation}
    a = \begin{pmatrix}
    a_1 & a_2 \\
    a_2^* & a_1^*
    \end{pmatrix} \; ,\quad 2V_{nd} = 2\begin{pmatrix}
    V_1 & V_2 \\
    -V_2^* & -V_1^*.
    \end{pmatrix}
\end{equation}
Then \cref{eq:a_e_commutator} is equivalent to the system:
\begin{equation}
    \left[E, a_1\right] = V_1 \; , \quad  \left\{E, a_2\right\} = V_2.
\end{equation}
Using the fact that \(E\) is diagonal we come to: 
\begin{equation}\label{eq:a_perturb_general}
\begin{aligned}
    (a_1)_{\alpha\beta} &= \frac{2(V_1)_{\alpha\beta}}{E_{\alpha} - E_{\beta}} \textup{ for } \alpha \neq \beta \\
    (a_2)_{\alpha\beta} &= \frac{2(V_2)_{\alpha\beta}}{E_{\alpha} + E_{\beta}}.
\end{aligned}
\end{equation}
From \cref{eq:disord_cond} \((V_1)_{\alpha\beta} , (V_2)_{\alpha\beta} = \Theta(N^{-3})\). The spectrum is gapped, thus \(E_{\alpha} + E_{\beta} = \Theta(1)\). For close energy levels with \(\abs{\alpha - \beta} \sim 1\), their difference \(E_{\alpha} - E_{\beta} = \Theta(N^{-1})\) if the spectrum has a non-zero derivative at \(E_\alpha \approx E_\beta\) as a function of momentum. If only the second derivative is non-zero, then \(E_{\alpha} - E_{\beta} = \Theta(N^{-2})\) -- this is the case at \(E_\alpha, E_\beta \approx E_{min} = E_N\), as we shall see (c.f. \cref{eq:e_diff_deriv}). Consequently, from  \cref{eq:a_perturb_general} we find asymptotics for the maximal values of corresponding matrices:
\begin{equation}
\begin{aligned}\label{eq:a_asymptot}
    \max\limits_{\alpha\beta} (a_1)_{\alpha\beta} &= \Theta(N^{-1}) \, \text{, achieved at } \alpha, \beta \approx N, \\ 
    \max\limits_{\alpha\beta} (a_2)_{\alpha\beta} &= \Theta(N^{-3}).
\end{aligned}
\end{equation}
The initial state of the fermionic system can be written in terms of the modes and the vacuum state corresponding to both the Hamiltonian without disorder and with it. That is:
\begin{equation}\label{eq:init_state_two_ways}
    \exp(\Tilde{\eta}_\alpha^\dagger \Tilde{G}_{\alpha\beta} \Tilde{\eta}_\beta^\dagger)\ket{vac}_{\Tilde{\eta}} = \exp({\eta}_\alpha^\dagger {G}_{\alpha\beta} {\eta}_\beta^\dagger)\ket{vac}_{\eta}.
\end{equation}
The matrix \(T\) transforming one vacuum to the other can be written as:
\begin{equation}
\begin{aligned} \label{eq:t_mat_expr}
     \exp(\Tilde{\eta}_\alpha^\dagger T_{\alpha\beta} \Tilde{\eta}_\beta^\dagger)\ket{vac}_{\Tilde{\eta}} =\ket{vac}_{\eta} \\  T = (1 + a_1)^{-1} a_2 = \mathcal{O}(N^{-3}).
\end{aligned}
\end{equation}
Below we use the following notation:
\begin{equation}
\begin{aligned} \label{eq:gmat_shorts}
    \mathcal{T} = \Tilde{\eta}_\alpha^\dagger T_{\alpha\beta} \Tilde{\eta}_\beta^\dagger \; , \; \mathcal{G} = {\eta}_\alpha^\dagger {G}_{\alpha\beta} {\eta}_\beta^\dagger \\  \mathcal{G}_0 = \Tilde{\eta}_\alpha^\dagger {G}_{\alpha\beta} \Tilde{\eta}_\beta^\dagger \; , \; \mathcal{G}_1 = \Tilde{\eta}_\alpha^\dagger {\Tilde{G}}_{\alpha\beta} \Tilde{\eta}_\beta^\dagger,
\end{aligned}
\end{equation}
which allows us to rewrite \cref{eq:init_state_two_ways} as:
\begin{equation}
    \exp(\mathcal{G}_1) = \exp(\mathcal{G})\exp(\mathcal{T}).
\end{equation}
Our task as to find the change in $\mathcal{G}$:
\begin{equation}
    \Delta \mathcal{G} = \mathcal{G}_1 - \mathcal{G}_0.
\end{equation}

Expressing the right hand side of \cref{eq:init_state_two_ways} in terms of the \(\Tilde{\eta} \) operators we obtain:

\begin{equation}\label{eq:g_transf_full_matrix}
    \eta_\alpha^\dagger {G}_{\alpha\beta} {\eta}_\beta^\dagger = \begin{pmatrix}
    \Tilde{\eta} \\
    \Tilde{\eta}^\dagger
    \end{pmatrix}^T \begin{pmatrix}
    a_2 \\ 1 + a_1
    \end{pmatrix} {G} \begin{pmatrix}
    -a_2^* & 1 - a_1^*
    \end{pmatrix}\begin{pmatrix}
    \Tilde{\eta} \\
    \Tilde{\eta}^\dagger
    \end{pmatrix}.
\end{equation}

Remembering \cref{eq:g_asymptot,eq:a_asymptot}, we extract the only part of \cref{eq:g_transf_full_matrix} that is non-vanishing with \(N \rightarrow \infty\):
\begin{gather}\label{eq:g_mat_appr_transform}
\mathcal{G} = \eta_\alpha^\dagger {G}_{\alpha\beta} {\eta}_\beta^\dagger = \Tilde{\eta}^\dagger {G} \Tilde{\eta}^\dagger + \Tilde{\eta}^\dagger G (-a_2)^*\Tilde{\eta} + \Tilde{\eta}a_2 G \Tilde{\eta}^\dagger + \\ \nonumber
\Tilde{\eta}^\dagger G (-a_1)^* \Tilde{\eta}^\dagger + \Tilde{\eta}^\dagger a_1 G \Tilde{\eta}^\dagger + \Tilde{\eta} a_2 G (-a_2^*) \Tilde{\eta} + \Tilde{\eta} a_2 G (-a_1^*)\Tilde{\eta}^\dagger + \\\nonumber
    \Tilde{\eta}^\dagger a_1 G (-a_2)^* \Tilde{\eta} + \Tilde{\eta}^\dagger a_1 G (-a_1)^* \Tilde{\eta}^\dagger \approx \mathcal{G}_0 + \\\nonumber
    \Tilde{\eta}^\dagger G (-a_1)^* \Tilde{\eta}^\dagger + \Tilde{\eta}^\dagger a_1 G \Tilde{\eta}^\dagger. 
\end{gather}
Using the Baker–Campbell–Hausdorff formula on the RHS of \cref{eq:init_state_two_ways}, we obtain in the exponent:
\begin{equation}\label{eq:bch_full}
\mathcal{G}+\left(\mathcal{T} + {\frac {1}{2}}[\mathcal{G},\mathcal{T}]+{\frac {1}{12}}[\mathcal{G},[\mathcal{G},\mathcal{T}]]-{\frac {1}{12}}[\mathcal{T},[\mathcal{G},\mathcal{T}]]+\cdots \right)
\end{equation}
In \cref{sec:commut_asympt} we show that the sum of the commutator series in brackets in \cref{eq:bch_full} scales as \(N^{-3}\), so now we neglect it. Then, substituting the approximation for \(\mathcal{G}\) from \cref{eq:g_mat_appr_transform} into \cref{eq:bch_full}, we arrive at:
\begin{equation}
\begin{aligned}\label{eq:equiv_on_vecs}
    \exp(\mathcal{G}_1)&\ket{vac}_{\Tilde{\eta}} \approx \\ 
    &\exp(\mathcal{G}_0 + \Tilde{\eta}^\dagger G (-a_1)^* \Tilde{\eta}^\dagger + \Tilde{\eta}^\dagger a_1 G \Tilde{\eta}^\dagger)\ket{vac}_{\Tilde{\eta}}.
\end{aligned}
\end{equation}
Comparing coefficients in front of the two-particle states in \cref{eq:equiv_on_vecs}, we obtain:
\begin{equation}\label{eq:gmat_corr_shrtfrm}
    \Delta \mathcal{G} = \Tilde{\eta}^\dagger G (-a_1)^* \Tilde{\eta}^\dagger + \Tilde{\eta}^\dagger a_1 G \Tilde{\eta}^\dagger + \mathcal{O}(N^{-2}),
\end{equation}
or in the matrix form:
\begin{equation}\label{eq:dG_matform}
\Delta G_{\alpha\beta} = \Tilde{G}_{\alpha\beta} - G_{\alpha\beta} = (a_1 G - G a_1^*)_{\alpha\beta} + \mathcal{O}(N^{-2}).
\end{equation}
From \cref{eq:a_asymptot} we remember \((a_1)_{\alpha\beta} = \Theta(N^{-1})\) and from \cref{eq:g_asymptot} \(\abs{G_{\alpha,\alpha+1}} = \Theta(N)\), both at the lowest energies (\(\alpha, \beta \approx N\)). Thus, \(\Delta G_{\alpha\beta} = \Theta(1)\) at the lowest energies. This means that for an infinitesimal change in the Hamiltonian \(\mathcal{V}\) we obtained a finite change in \(G_{\alpha\beta}\) for \(\alpha, \beta \approx N\). Additionally, in the second statement of \cref{theorem:main} we claimed that the initially non-zero entries of \(G\) do not change in the thermodynamic limit. This follows from \cref{eq:dG_matform} and the absence of diagonal entries in matrix \(a_1\). The latter is a consequence of \cref{eq:a_perturb_general} and the definition of \(V_1\) as the block of \textit{off-diagonal} part of the perturbation \(\mathcal{V}\). Thus, the second statement of \cref{theorem:main} is proved.

Finally, from \cref{eq:energy_change} and the asymptotics we chose for \(\mathcal{V}\) in \cref{eq:disord_cond}, it follows that the change in the eigenenergies vanishes in the \(N \rightarrow \infty\) limit. This proves the first statement of \cref{theorem:main} and finishes the proof.
\end{proof}

\subsection{Application of \cref{theorem:main} to Harmonic Perturbation}\label{sec:appl_harmonic}
In \cref{subsec:gen_expl} we gave a qualitative picture of the change in Fisher zeros for any weak perturbation \(\mathcal{V}\). In this section, we work with a specially chosen \(\mathcal{V}\) to obtain some quantitative results. We consider a perturbation of the form:
\begin{equation}\label{eq:disord_q_def}
    \mathcal{V} = 
    \sum_{j \in [-N/2\dots N/2]} \sum_{\substack{q \in BZ \\q = \mathcal{O}(N^{-1})}} b_q \cos(q j) \sigma_j^z,
\end{equation}
where we chose to sum only over small momenta in the Brillouin zone. This is equivalent to the following perturbation in the fermionic model:
\begin{equation}
    \mathcal{V} = \sum_{j, q} 2 b_q \cos(q j) c_j^\dagger c_j.
\end{equation}
The fourier transform of a single mode is:
\begin{equation} \label{eq:v_fourier}
    \mathcal{V}_q = b_q \sum_k c_k^\dagger c_{k+q} + c_k^\dagger c_{k-q} = \sum_k \mathcal{V}_{k,q},
\end{equation}
or in terms of the Bogolyubov operators:
\begin{equation}
    c_k = \cos(\theta_k / 2) \eta_k + i \sin(\theta_k / 2) \eta_{-k}^\dagger,
\end{equation}
which upon substitution to \cref{eq:v_fourier} give us:
\begin{align} \nonumber
    &\frac{\mathcal{V}_{k,q}}{b_q} = (\cos(\theta_k / 2) \eta_k^\dagger - i \sin(\theta_k / 2) \eta_{-k})(\cos(\theta_{k+q} / 2) \eta_{k+q}\\
    \nonumber
   &+i \sin(\theta_{k+q} / 2)\eta_{-k-q}^\dagger) +  (\cos(\theta_k / 2) \eta_k^\dagger - i \sin(\theta_k / 2) \eta_{-k})\\
    &\times(\cos(\theta_{k-q} / 2) \eta_{k-q} + i \sin(\theta_{k-q} / 2) \eta_{-k+q}^\dagger).
    \label{eq:v_bogolyub}
\end{align}
The matrix elements of the perturbation in the energy basis \(V_{\alpha\beta}\) are non-zero only for energies separated by a momentum \(q\). Also, we require that \(q = \mathcal{O}(N^{-1})\), therefore our perturbation fulfills the scaling for \(a_1\) from the \cref{eq:a_asymptot}. Thus, we may use the approximate expression for corrections to \(G\) shown in \cref{eq:dG_matform}:
\begin{equation}
    \Delta{G} \approx G (-a_1)^* + a_1 G.
\end{equation}
Using \cref{eq:a_perturb_general} and the structure of \(G\) in the eigenbasis of \(\ham_1\):
\begin{equation}
\begin{aligned}\label{eq:delta_g_main}
    \Delta G_{\alpha\beta} \approx G_{\alpha \pm 1, \alpha} (-a_1)^*_{\alpha\beta} + (a_1)_{ij} G_{\beta, \beta \pm 1} \\ = 2\frac{(V_1)_{\alpha\beta} G_{\beta, \beta \pm 1} - (V_1)_{\alpha\beta}^* G_{\alpha, \alpha \pm 1}}{E_\alpha - E_\beta}, 
\end{aligned}
\end{equation}
where \(G_{\alpha, \alpha \pm 1}\) denotes \(G_{\alpha, \alpha+1}\) if \(i\) is odd and \(G_{\alpha, \alpha - 1}\) if \(i\) is even (see \cref{sec:g_shape}). Now we will derive the explicit form of the asymptotics of the \cref{eq:delta_g_main} for \(k \rightarrow \pi\) and the initial field \(h^0 = 0\) and post-quench field with a mean value \(h^1 = h\): 
\begin{align}\label{eq:g_sq_exact}
    \abs{G_{k,-k}}^2 &= \frac{1 - \cos(\theta_k - \Tilde{\theta}_k)}{1 + \cos(\theta_k - \Tilde{\theta}_k)}, \tan \theta_k = \frac{\sin k }{h + \cos k},\\
    \label{eq:theta_approx}
    \theta_k &\approx \frac{k - \pi}{1 - h}      \theta_k - \Tilde{\theta}_k \approx \pi + (k - \pi) \frac{h}{h - 1},\\
    \label{eq:g_diag_exact}
    \abs{G_{k,-k}}^2 &\approx \frac{1 + \cos((k - \pi) \frac{h}{h - 1})}{1 - \cos((k - \pi) \frac{h}{h - 1})} \approx \frac{4}{\left((k - \pi) \frac{h}{h - 1}\right)^2}.
\end{align}
For \(h > 1\) and \(k < \pi\) this gives:
\begin{equation}\label{eq:g_approx_deriv}
    \abs{G_{\alpha,\alpha+1}} = \abs{G_{\alpha+1,\alpha}} = \frac{2(h-1) / h}{\pi - k_\alpha} = \Theta(N),
\end{equation}
\begin{figure*}[!htb]
  \centering
    \includegraphics[width=0.8\textwidth]{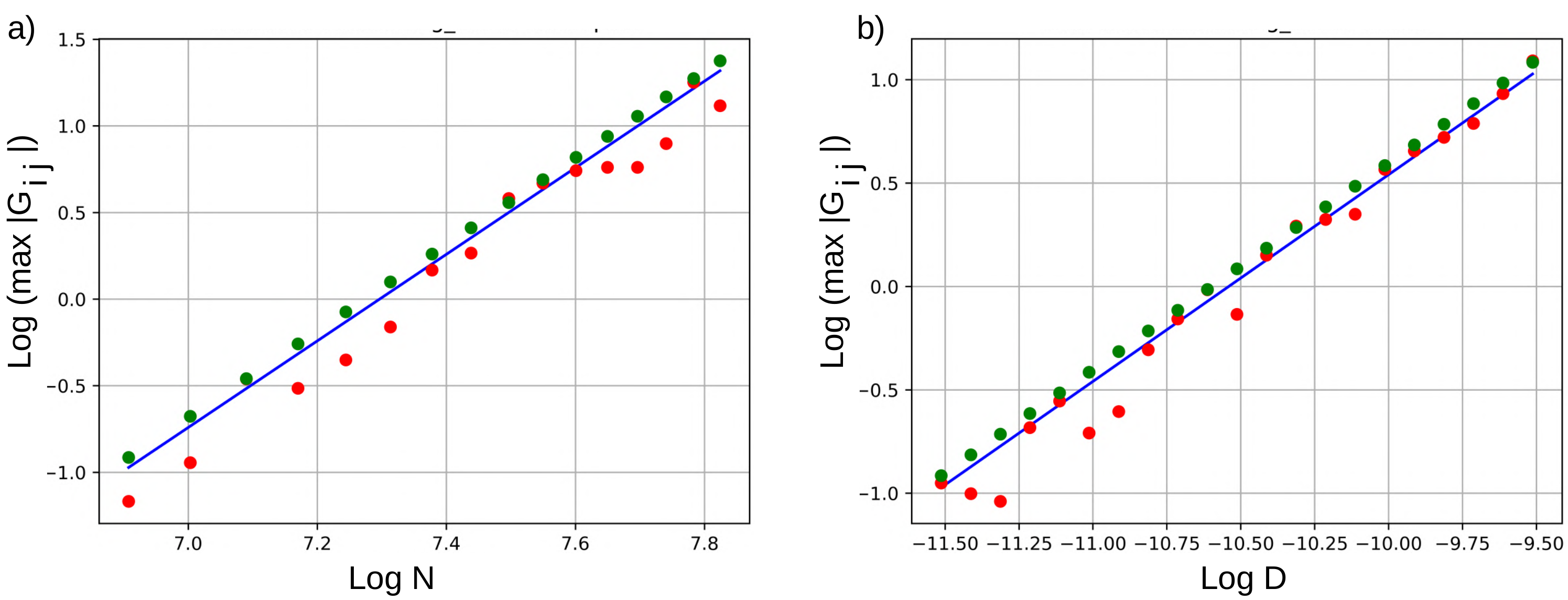}
    \caption{In both pictures \textbf{(a)} and \textbf{(b)} on the vertical axis is the change of BCS matrix elements \(G_{N, N-2}\) corresponding to the next to lowest energy Cooper pairs (lowest energy pairs formed by excitations with different energies) in logarithmic scale. Such Cooper pairs are the most sensitive to perturbation in the external magnetic field. In \cref{sec:appl_harmonic}, in particular \cref{eq:gij_final_momentum} we show that these are pairs where one excitation has the lowest energy in the spectrum and the other excitation - next to the lowest energy. The plot \textbf{(a)} shows how the response in the quantity of such pairs scales with the length of the spin chain \(N\), the \textbf{(b)} plot shows how it scales with the amplitude \(D\) of perturbation in the external magnetic field. Blue line shows theoretical prediction \cref{eq:gij_numerical_general} for response to a sinusoidal perturbation, see \cref{eq:disord_q_def} where only the lowest momentum \(q = q_{min}\) perturbation amplitude is non-zero \(b_q \neq 0\) and equal to \(b_q = D\). Green dots show numerically obtained result for the same quantity. Red dots show numerical result for the same quantity, but when when perturbations on other frequencies are random instead of being zero. That means, for \(q = q_{min}\) \(b_q = D\), as earlier, but for other \(q\), \(b_q\) are no longer zero, instead, they are random: \(b_q \in \mathcal{U}_{[-D, D]}\).} \label{fig:gij_numerical_disord}
\end{figure*} 
where since \(k_\alpha\) is close to \(\pi\) we have introduced \(q_\alpha = \pi - k_\alpha\) of an order of several steps in the Brillouin zone, that is \(q_\alpha = \mathcal{O}(N^{-1})\). Similarly, we put \(q_\beta = \pi + k_\beta = \mathcal{O}(N^{-1})\) and let \(E_i = E(k_\beta) , \; E_\beta = E(k_\beta)\), then we can expand
\begin{equation}\label{eq:e_thr_momentum}
    E(k) = \sqrt{h^2 + 2h\cos(k) + 1}
\end{equation}
near \(\pi\), \(-\pi\) and use it to calculate the energy difference. We expand to the second order because the first derivative at \(\pi\), \(-\pi\) is zero:
\begin{equation}\label{eq:e_diff_deriv}
    E_\alpha - E_\beta = \frac{h}{2(h - 1)}(q_\alpha^2 - q_\beta^2) = \mathcal{O}(N^{-2}).
\end{equation}
Here we took \(h > 1\). From \cref{eq:v_bogolyub} we derive:
\begin{equation}\label{eq:v_final_approx}
    (V_1)_{\alpha\beta} = (V_1)_{\alpha\beta}^* = 2b_q \cos(\frac{\theta_{k_\alpha}}{2})\cos(\frac{\theta_{k_\beta}}{2}) \approx 2 b_q.
\end{equation}
Substituting \cref{eq:g_approx_deriv,eq:e_diff_deriv,eq:v_final_approx} into \cref{eq:delta_g_main}, we finally obtain:
\begin{equation}
\begin{aligned}\label{eq:gij_first_ord}
    \Delta G_{\alpha\beta} \approx 2(2b_q) \left(\frac{2(h-1)}{h}\right)^2\left(\frac{\frac{1}{q_\alpha} - \frac{1}{q_\beta}}{q_\alpha^2 - q_\beta^2}\right)\\=
    -\omega \frac{1}{(q_\alpha + q_\beta)q_\alpha q_\beta} \; , \; \omega = 4b_q \left(\frac{2(h-1)}{h}\right)^2.
\end{aligned}
\end{equation}
With perturbation wave vector \(q = q_\beta - q_\alpha\). In the derivation we used several assumptions:
\begin{enumerate}
    \item Thermodynamic limit \(N \gg 1\)
    \item \(k_\alpha\), \(k_\beta\) are close to \(\pi\), \(-\pi\) respectively; or equivalently, the energies \(E_\alpha, \; E_\beta\) are close to the lowest end of the spectrum
    \item The potential \(\mathcal{V}\) has a single mode \(q\) (see \cref{eq:v_bogolyub}) connecting energies \(E_\alpha, \; E_\beta\). In the case of several such modes \(q_s \,\), there will be a \(\sum\limits_{q_s} b_{q_s}\) in the numerator of \(\alpha\) \cref{eq:gij_first_ord}
    \item The potential has critical scaling (\(\nu = \Theta(N^{-3})\), see \cref{eq:disord_cond}).
\end{enumerate}
Note, that only  for \(\alpha = \beta \pm 1\) (\(+\) for \(\alpha\) odd, \(-\) for \(\alpha\) even) \(G_{\alpha\beta} \neq 0\) (see \cref{sec:g_shape}). Therefore, for all the other elements \(\Delta G_{\alpha\beta} = \Tilde{G}_{\alpha\beta}\).
The possible range of momentum values is \(k = \frac{2\pi}{N} p \, , \, p \in (-\frac{N-2}{2}, \dots, 0, \dots, \frac{N}{2})\). We designate \(k_\alpha = \pi - \frac{2\pi}{N}p_\alpha \, , \; k_\beta = -\pi + \frac{2\pi}{N}p_\beta\) and obtain:
\begin{equation}\label{eq:gij_final_momentum}
    \Tilde{G}_{\alpha\beta} = \Tilde{G}_{k, -k+q} \approx \omega' N^3 \frac{1}{(p_\alpha + p_\beta) p_\alpha p_\beta},
\end{equation}
where \(p_\alpha, \, p_\beta\) are arbitrary integers such that \(p_\alpha, \, p_\beta \ll N\), and \(\omega' = - \frac{4b_q}{(2 \pi)^3} \left(\frac{2(h-1)}{h}\right)^2\). Consequently, for the maximal entry of the matrix \(\Delta G\) in the momentum basis we have to choose \(p_\alpha = 1\),\ \(p_\beta = 2\), \(q = 2\pi / N\) and obtain:
\begin{equation}\label{eq:gij_numerical_general}
    \log (\max_{k_1, k_2} \abs{\Delta G_{k1, k2}}) = 3 \log(N) + \log(b_q) + \const{}.
\end{equation}
From \cref{eq:gij_numerical_general} follows the condition on the minimal modulation amplitude necessary to cause a dynamical quantum phase transition. Namely, for the perturbation-induced zeros to cross the imaginary axis, the maximal perturbation-induced element of matrix \(G\) must satisfy 
\begin{align}
    \log (\max\limits_{k_1 \neq -k_2} \abs{\Delta G_{k1, k2}}^2) \geq 0,
\end{align}
see \cref{eq:zeros_calc}. Taking the borderline case we obtain:
\begin{equation}
\begin{aligned}\label{eq:mindis}
    \log(b_q)_{min} = - 3 \log(N) + \const{}\\
    (b_q)_{min} \geq \frac{3(2\pi)^3}{8} \left(\frac{h}{h - 1}\right)^2 \frac{1}{N^3} .
\end{aligned}
\end{equation}
In the position basis the coefficient in front of \(\log(N)\) in \cref{eq:gij_numerical_general} will change to \(2.5\) due to the \(1 / \sqrt{N}\) normalization of the Fourier transform. Additionally, if we choose a random perturbation instead of a sinusoidal one, higher modes will become non-zero. This will lead to Fisher zero points being scattered around the line described in \cref{eq:gij_numerical_general} - see \cref{fig:gij_numerical_disord} and \cref{sec:dgij_num_test} for the details of the corresponding numerical calculation.

\section{Influence of Disorder on the Fermionic Correlators}\label{sec:theory_correls}
In \cref{sec:main_res} we gave an intuitive explanation for the insensitivity of correlators to disorder-induced DQPTs. Now we make a precise statement. The proof operates only with the BCS matrix \(G\).
\begin{theorem}\label{theorem:correl}
Suppose the following holds:
\begin{enumerate}
    \item We introduce perturbations with only long-wavelength (\(q = \mathcal{O}(N^{-1})\)) Fourier components of order \(V_q = \mathcal{O}(N^{-3})\) and all other Fourier components zero. Note, that in \cref{sec:theory_zeros} we showed that this is sufficient to induce a second series of DQPTs. \label{cond:correl_th_dis_type}
    \item We consider two spins at a distance \(\abs{i-j} = d \ll N\) which is fixed and independent of the total number of spins \(N\). \label{enum:close_spins}
\end{enumerate}
Then for these two spins correlators in the \(x\)-direction change, compared to the homogeneous case, as:
\begin{equation}
    \Delta \langle \sigma_i^x \sigma_j^x \rangle = \mathcal{O}(N^{-1}).
    \tag{\ref{eq:x_correl_change}}
\end{equation}
Only the first condition suffices to ensure that \(z\) correlators vanish:
\begin{equation} \label{eq:z_correl_change}
    \Delta \langle \sigma_i^z \sigma_j^z \rangle = \mathcal{O}(N^{-1}).
\end{equation}
\end{theorem}
\begin{proof}
From \cref{eq:corrx_pfaff} and the note below it, it is clear that the spin-spin correlators are polynomials in fermionic correlators. In the case of \(\langle \sigma_i^z \sigma_j^z \rangle\) correlators, the degree of such a polynomial is always 2, while for \(\langle \sigma_i^x \sigma_j^x \rangle\) correlators under condition $2$ of the theorem, the degree is \(\abs{i-j} \ll N\) and independent of \(N\). Thus, it is enough to prove that the change in all fermionic correlators scales as \(\mathcal{O}(N^{-1})\) to prove the theorem (for our purposes any negative power of \(N\) would suffice). 

The proof is technical and conducted in \cref{sec:correl_bounds}. The main idea is to express fermionic correlators through the BCS-matrix \(G\) and then substitute disorder-induced corrections from \cref{eq:dG_matform}.
\end{proof}
\begin{itemize}
\item [\textit{Notes:}]
\item Though formally we only proved scaling \(\mathcal{O}(N^{-1})\) for the change in arbitrary fermionic correlators and only short-range spin-spin correlators, numerical simulations show that the same is also true for long-range spin-spin correlators -- see \cref{fig:fisher_zeros_no_small,fig:strong_dis}, where correlators shown are for spins separated by half the chain's length.
\item The proof relies on the calculations from \cref{sec:correl_bounds}, which utilize the fact that \(\abs{G_{\alpha,\alpha+1}}, \abs{G_{\beta, \beta+1}} \gg \abs{\Delta G_{\alpha\beta}}\). Because the \(G\)-matrix determines the coordinates of the Fisher zeros, we can say that the behavior of spin-spin correlators, unlike the behavior of Loschmidt echo, is not determined by the local properties of Fisher zeros near the crossing of the imaginary axis. Instead, to calculate the effect of DQPTs on correlators we need information about all the Fisher zeros.
\end{itemize}
This theorem can be generalized to \(L\)-spin correlators. The only non-zero spin correlators along the \(x\)-axis are those with an even number of spins. Indeed, the Jordan-Wigner transform turns all \(\sigma_i^x\) into an odd number of spin operators, thus correlators of an odd number of spins result in correlators of an odd number of fermions. These are all zero, because we started with a state with an even number of fermions, and evolved it with a parity-preserving Hamiltonian.

Suppose that in an \(L\)-chain of spin operators \(\sigma_l^x\) is the leftmost spin, \(\sigma_r^x\) is the rightmost spin. Also suppose \(\abs{r-l} \ll N\) and \(L \ll N\) and both are independent of \(N\). Then we can similarly apply the Jordan-Wigner transformation to each spin operator in the chain: 

\begin{equation}
\sigma_i^x = c_i^\dagger \exp(-i\pi \sum\limits_{j=1}^{j=i-1} n_j) + c_i \exp(i\pi \sum\limits_{j=1}^{j=i-1} n_j), 
\end{equation}
where \(n_j = c_j^\dagger c_j\) is the number of excitations on \(j\)-th site.

All the contributions to the phase \(\exp(i\pi n_j)\) with \(j < l\) (\(j > r\)) commute with all the other operators and can be pulled from each \(\sigma_i^x\) to the very left (to the very right) of the operator chain respectively. For an even number of spin operators \(L\), each such phase contribution \(\exp(i\pi n_j)\) has an even power and thus is canceled. We are left with \(L\) operators of the form \(c_i^\dagger + c_i\) and at most \(\abs{r-l}\) operators of the form \(\exp(i\pi n_k)\). This is a sum of \(2^L\) chains of at most \(m = L + 2\abs{r-l}\) fermionic operators, and to each chain we can apply Wick's theorem. Since \(m \ll N\) and independent of \(N\) we can place an upper bound on the change of each chain as in \cref{eq:ferm_scale_intro} with \(\mathcal{O}(N^{-1})\). Finally, we use that \(L\) is independent of \(N\) and additionally suppose that \(2^L \ll N\) to say that the sum of all the chains can be bound with \(2^L \mathcal{O}(N^{-1})\). The final result is \cref{eq:ferm_scale_intro_many}:
\begin{equation}\tag{\ref{eq:ferm_scale_intro_many}}
    \Delta \underbrace{\langle \sigma_{i1}^x \sigma_{i2}^x \dots \sigma_{iL}^x \rangle}_{L \text{ spins}} = \mathcal{O}(N^{-1})
\end{equation}
For \(z\)-spin correlators the generalization from two spin case to \(n\) spin case is even more direct, because \(\sigma_i^z = 1 - 2 n_i\) and \(n_i\) all commute with each other.

\section{Conclusion}\label{sec:conclusion}

We demonstrated that the disorder-induced dynamical quantum phase transition in the Ising model is an example of a non-topological dynamical quantum phase transition without a local order parameter. A series of critical times universally appears for any vanishing perturbation with Fourier components at the lowest momentum of order \(1/N^3\). This could be considered a dynamical counterpart of the Anderson orthogonality catastrophe \cite{anderson1967infrared}. That is, a vanishingly small perturbation causes a large deviation in the many-body wave function, while the observables remain intact. In our setting, it is the Loschmidt echo that changes drastically.

Several intriguing questions can be addressed in future studies. As we have an example of the DQPT where no order parameter can be found, a natural question is whether this phase transition belongs to a larger class with the same property. Vice versa: what class of the DQPTs can be endowed with an order parameter?  

\section{Acknowledgments}
The authors are grateful to Peru d'Ornellas, who read carefully and helped greatly to edit the first version of the manuscript. The work was carried out in the framework of the Roadmap for Quantum computing in Russia.


%
\newpage
\appendix
\section{Shape of BCS-matrix \(G\) without disorder}\label{sec:g_shape}
With homogeneous external field BCS-matrix \(G\), determining wave function as in \cref{eq:psi_t_full}, has the following form in the eigenbasis of \(\ham_1\):
\begin{equation*}
\begin{pmatrix}
\ddots & \vdots & \vdots & \vdots & \vdots & \vdots \\
\dots & 0 & 0 & 0 & 0 & 0 \\
\dots & 0 & 0 & -G_{N-2, N-3} & 0 & 0 \\
\dots & 0 & G_{N-2, N-3} & 0 & 0 & 0 \\
\dots & 0 & 0 & 0 & 0 & -G_{N, N-1} \\
\dots & 0 & 0 & 0 & G_{N, N-1} & 0 \\
\end{pmatrix}
\end{equation*}
In a homogeneous external field modes with momenta differing by sign (\(k\) and \(-k\) for all \(k\)) have the same energy. This leads to the degeneracy \(E_\alpha = E_{\alpha-1}\) for even indices \(\alpha\). This means that in this homogeneous case matrix \(G\) only pairs excitations of the same energy. That is, we have only Cooper pairs with both components of the same energy and opposite momenta. For example, \(G_{N, N-1}\) corresponds to \(\eta_{N}^\dagger \eta_{N-1}^\dagger\) pairs of excitations with energies \(E_{N-1} = E_N = E_{min}\), \(G_{N-2, N-3}\) corresponds to next-to-lowest energy pairs of excitations and so on.
\section{Asymptotics of commutator series}\label{sec:commut_asympt}
Here we prove that the sum of commutators (henceforth \(S_c\)) in the second brackets of \cref{eq:bch_full} has asymptotics \(\mathcal{O}(N^{-3})\) for \(N \rightarrow \infty\) and thus vanishes in thermodynamic limit. First, we look at a chain of commutators of length \(k\):
\begin{equation}\label{eq:k_comm_chain}
    \underbrace{[\mathcal{G}, [\mathcal{G}, [\dots [\mathcal{T}, [\mathcal{G}, \mathcal{T}]]]]]}_{k \text{ commutators}}
\end{equation}
There are \(2^k\) of such chains in \(S_c\) (some are zero). Also, each one enters \(S_c\) with a coefficient whose absolute value is less than 1. Therefore, we can bound the sum of length-\(k\) chains of commutators with \(2^k M\), where \(M\) is the upper bound for one such chain. To obtain M we will group all elements of \(\mathcal{G}\) and \(T\) by \textit{types}: \(\eta \eta\) will stand for \(\eta_\alpha \eta_\beta\) with any \(\alpha, \beta\), \(\eta^\dagger \eta^\dagger\) - for \(\eta_\alpha^\dagger \eta_\beta^\dagger\) and \(\eta^\dagger \eta\) - for both \(\eta_\alpha^\dagger \eta_\beta\) or \(\eta_\alpha \eta_\beta^\dagger\). Note that a commutator of elements of two \textit{types} again belongs to one of the \textit{types} (or is \(0\)), as shown in \cref{tab:ferm_pair_comm_short} (for a detailed derivation see \cref{tab:ferm_pair_comm}). We will track how the scalar prefactor in front of elements of each type and with each pair of indices changes as we consecutively compute all \(k\) correlators in a chain. We start with \(\mathcal{T}\), whose all elements are of \(\eta^\dagger \eta^\dagger\) \textit{type}. With any pair of indices \(\alpha, \beta\) elements \(\eta_\alpha^\dagger \eta_\beta^\dagger\) in \(\mathcal{T}\) have a prefactor scaling as \(\mathcal{O}(N^{-3})\).  
\begin{table}[htb!]
\caption{At the intersection of a column titled with a type A and a row titled with a type B is a type of \([A, B]\)}
\label{tab:ferm_pair_comm_short}
\begin{center}
\begin{tabular}{ |c||c|c|c| } 
 \hline
 \phantom                      & \(\eta \eta\)         & \(\eta^\dagger \eta\)         & \(\eta^\dagger \eta^\dagger\) \\ \hline \hline
 \(\eta \eta\)                 & 0                     & \(\eta \eta\)                 & \(\eta^\dagger \eta\) \\ \hline
 \(\eta^\dagger \eta\)         & \(\eta \eta\)         & \(\eta^\dagger \eta\)         & \(\eta^\dagger \eta^\dagger\) \\ \hline
 \(\eta^\dagger \eta^\dagger\) & \(\eta \eta^\dagger\) & \(\eta^\dagger \eta^\dagger\) & 0 \\
 \hline
\end{tabular}
\end{center}
\end{table}
Further, we analyze \cref{eq:g_mat_appr_transform}, to understand, what \textit{types} enter in \(\mathcal{G}\) and how their prefactors scale. First, matrix \(G\) elements scale at most as \(\mathcal{O}(N)\) (see \cref{eq:g_approx_deriv}). Second, each element of \(a_1\) scales at most as \(\mathcal{O}(N^{-1})\);  each element of \(a_2\) scales at most as \(\mathcal{O}(N^{-3})\)  - see \cref{eq:a_asymptot}. Now we can bound from above the scaling of products of \(G, a_1, a_2\). Elements of \(a_1 G, G a_1\) scale at most as \(\mathcal{O}(N) \cdot \mathcal{O}(N^{-1}) = \mathcal{O}(1)\) and \(a_2 G, G a_2\) as \(\mathcal{O}(N) \cdot \mathcal{O}(N^{-3}) = \mathcal{O}(N^{-2})\), because \(G\) has only one element per each row/column so that each element of \(a_1 G, G a_1\) is just a product of one element of \(G\) and one element of \(a_1\) (similarly for \(a_2\)). On the contrary, to bound from above scaling of matrices \(a_1 G a_1, a_2 G a_1, a_1 G a_2, a_2 G a_2\) we need to multiply scaling functions of corresponding matrices and additionally multiply the result by \(N\) because matrix multiplication of two generic matrices leads to each element of the resulting matrix being a sum of \(N\) products of elements of initial matrices. Consequently, we have scalings \(a_1 G a_1 - \mathcal{O}(1); \, a_2 G a_1, a_1 G a_2 - \mathcal{O}({N^{-2}}); \, a_2 G a_2 - \mathcal{O}({N^{-4}})\). To sum up, \(\mathcal{G}\) consists of summands of:

\begin{enumerate}
    \item \textit{Type} \(\eta^\dagger \eta^\dagger\) - special (one element per row/column) matrix \(G\), scaling as \(\mathcal{O}(N)\) and generic matrices \(Ga_1, a_1 G, a_1 G a_1\) scaling as \(\mathcal{O}(1)\)
    \item \textit{Type} \(\eta^\dagger \eta\) - generic matrices \(Ga _2, a_2 G, a_2 G a_1, a_1 G a_2\), scaling as \(\mathcal{O}(N^{-2})\)
    \item \textit{Type} \(\eta \eta\) - generic matrix \(a_2 G a_2\) scaling as \(\mathcal{O}(N^{-4})\).
\end{enumerate}

From \cref{tab:ferm_pair_comm} we see that for two pairs of fermionic operators to have a non-zero commutator they must have at least one element with a common index. Thus, for example, bound on asymptotics of \([A_{\alpha\beta} \eta_\alpha^\dagger \eta_\beta, B_{\gamma\delta} \eta_\gamma^\dagger \eta_\delta]\) is a product of asymptotics of elements of the matrix \(A\) times asymptotics of elements of the matrix \(B\) and times \(\mathcal{O}(N)\) (for each element \(A_{\alpha\beta} \eta_\alpha^\dagger \eta_\beta\) there are \(\mathcal{O}(N)\) elements of \(B_{\beta\gamma} \eta_\beta^\dagger \eta_\gamma\) with at least one common index). The same reasoning applies to other \textit{types}, except when \(A\) or \(B\) is \(G\), which is diagonal. In that case bound on asymptotics of \([A_{\alpha\beta} \eta_\alpha^\dagger \eta_\beta, B_{\gamma\delta} \eta_\gamma^\dagger \eta_\delta]\) is just a product of asymptotics of elements of \(A\) times asymptotics of elements of \(B\). \\

Combining all of the above we come to:

\begin{equation}\label{eq:type_transitions}
    \begin{cases}
    \eta^\dagger \eta^\dagger \xrightarrow[\eta^\dagger\eta^\dagger]{N^{+1}} 0 \\
    \eta^\dagger \eta^\dagger \xrightarrow[\eta^\dagger\eta]{N^{-1}} \eta^\dagger \eta^\dagger \\
    \eta^\dagger \eta^\dagger \xrightarrow[\eta\eta]{N^{-3}} \eta^\dagger \eta \\
    \end{cases}
    \begin{cases}
    \eta^\dagger \eta \xrightarrow[\eta^\dagger\eta^\dagger]{N^{+1}} \eta^\dagger \eta^\dagger \\
    \eta^\dagger \eta \xrightarrow[\eta^\dagger\eta]{N^{-1}} \eta^\dagger \eta \\
    \eta^\dagger \eta \xrightarrow[\eta\eta]{N^{-3}} \eta \eta \\
    \end{cases}
    \begin{cases}
    \eta \eta \xrightarrow[\eta^\dagger\eta^\dagger]{N^{+1}} \eta^\dagger \eta \\
    \eta \eta \xrightarrow[\eta^\dagger\eta]{N^{-1}} \eta \eta \\
    \eta \eta \xrightarrow[\eta\eta]{N^{-3}} 0 \\
    \end{cases}
\end{equation}
On the left side of an arrow we have a \textit{type} with which we start, under the arrow we have a \textit{type} with which we commute the first \textit{type}, above the arrow - a scalar prefactor we gain after the commutation, and on the right side of the arrow - the \textit{type} resulting from the commutation. To illustrate the use of \cref{eq:type_transitions}, suppose after \(k\) commutations we have \(A_{\alpha\beta} \eta_\alpha^\dagger \eta_\beta^\dagger\) and for all \(\alpha, \beta\) \(A_{\alpha\beta} \in \mathcal{O}(N^p)\). After the next commutation with \(\mathcal{G}\) we obtain
\begin{equation}\label{eq:type_transitions_example}
    A_{\alpha\beta} \eta_\alpha^\dagger \eta_\beta^\dagger \rightarrow B_{\alpha\beta} \eta_\alpha^\dagger \eta_\beta^\dagger + C_{\alpha\beta} \eta_\alpha^\dagger \eta_\beta
\end{equation} where for all \(\alpha, \beta\) \(B_{\alpha\beta} \in \mathcal{O}(N^{p-1})\) and \(C_{\alpha\beta} \in \mathcal{O}(N^{p-3})\). 
This information can be presented in the form of a finite state machine:
\begin{tikzpicture}[shorten >=1pt,node distance=3.5cm,auto, initial text = \(N^{-3}\)]
  \tikzstyle{every state}=[fill={rgb:black,1;white,10}]

  \node[state,initial]   (dd)                      {\(\eta^\dagger\eta^\dagger\)};
  \node[state] (dn) [below left of=dd]  {\(\eta^\dagger \eta\)};
  \node[state,accepting]           (zero) [below right of=dn]     {\(0\)};
  \node[state] (nn) [below right of=dd] {\(\eta\eta\)};

  \path[->]
  (dd) edge [bend right] node {\(N^{-3}\)} (dn)
       edge              node {} (zero)
       edge [loop right] node {\(N^{-1}\)} (   )
  (dn) edge [loop left]  node {\(N^{-1}\)} (   )
       edge [bend right] node {\(N^{+1}\)} (dd)
       edge [bend right] node {\(N^{-3}\)} (nn)
  (nn) edge [loop right] node {\(N^{-1}\)} (   )
       edge [bend right] node {\(N^{+1}\)} (dn)
       edge              node {}  (zero);
\end{tikzpicture}
Here we start with the topmost state \(\eta^\dagger \eta^\dagger\) scaling as \(\mathcal{O}(N^{-3})\) (these are elements of \(\mathcal{T}\)). The machine has one arrow cycles, changing asymptotics by \(\mathcal{O}(N^{-1})\), two arrow cycles, changing asymptotics by \(\mathcal{O}(N^{-2})\) and their combinations.Therefore, after \(k\) steps we either come to \(0\) or acquire prefactor at most \(\mathcal{O}(N^{-3}) \cdot \mathcal{O}(N^{-k})\). Hence, contribution of each one length-\(k\) chain to a coefficient in front of \(\eta^{(\dagger)}_i \eta^{(\dagger)}_j\), with \textit{type} and indices \(i, j\) fixed, is of order \(\mathcal{O}(N^{-k-3})\). After each commutation calculation may split into \(2\) or \(3\) branches (see \cref{eq:type_transitions,eq:type_transitions_example}), and overall there are \(2^{k}\) chains (see \cref{eq:k_comm_chain}). Therefore, contribution of all length-\(k\) chains  \cref{eq:k_comm_chain} to \(S_c\) is of order \(6^k \mathcal{O}(N^{-k-3})\). Summing over chain lengths from \(0\) to \(\infty\) we bound scaling of \(S_c\) with \(\mathcal{O}(N^{-3})\), which is what we wanted.
\section{Bounds on change of fermionic correlators}\label{sec:correl_bounds}
First, we will prove that introduction of disorder in the transverse fields changes all pairwise fermionic correlators in the energy basis very little, that is:
\begin{equation}
    \forall \alpha, \beta \; \Delta \langle \eta_\alpha \eta_\beta \rangle, \Delta \langle \eta_\alpha^\dagger \eta_\beta \rangle, \Delta \langle \eta_\alpha \eta_\beta^\dagger \rangle, \Delta \langle \eta_\alpha^\dagger \eta_\beta^\dagger \rangle = \mathcal{O}(N^{-1}).
\end{equation}
We will do that, expressing them through BCS-matrix \(G\) and employing our results for \(\Delta G_{\alpha\beta}\).
After that, as a corollary, we will obtain similar scaling results for change in correlators in \textit{position basis}. These are used in the proof of \cref{theorem:correl}.
\subsection{Fermions in energy basis}
We define:
\begin{equation} \label{eq:gamma_full_def}
    \Gamma = \begin{pmatrix}
    \Gamma_1 & \Gamma_2 \\
    \Gamma_3 & \Gamma_4
    \end{pmatrix},
\end{equation}
where \(\Gamma_{\alpha\beta}\) are fermionic correlators, whose expressions through \(G\) are derived in \cite{porro2022off}:
\begin{equation}\label{eq:etas_st_structure}
\begin{cases}
    (\Gamma_1)_{\alpha\beta} = (1 + \Tilde G^\dagger \Tilde G)_{\alpha\beta}^{-1} = \langle\eta_\alpha \eta_\beta^\dagger\rangle 
\\
    (\Gamma_2)_{\alpha\beta} = \Tilde G_{\alpha\gamma} (1 + \Tilde G^\dagger \Tilde G)_{\gamma\beta}^{-1} = \langle\eta_\alpha \eta_\beta\rangle 
\\
    (\Gamma_3)_{\alpha\beta} = (1 + \Tilde G^\dagger \Tilde G)_{\alpha\gamma}^{-1} \Tilde G_{\gamma\beta}^\dagger = \langle\eta_\alpha^\dagger \eta_\beta^\dagger\rangle 
\\
    (\Gamma_4)_{\alpha\beta} = \Tilde G_{\alpha\gamma} (1 + \Tilde G^\dagger \Tilde G)_{\gamma\delta}^{-1} \Tilde G_{\delta\beta}^\dagger = \langle\eta_\alpha^\dagger \eta_\beta\rangle 
\end{cases}
\end{equation}
Now, using \cref{eq:dG_matform} for the change of BCS-matrix \(G\) we aim to find scaling of \(\Delta \Gamma_{\alpha\beta}\). Designating the perturbed matrix \(G\) as \(\Tilde G = G + \Delta G\) we can obtain for the corrections to \(\Gamma_1, \Gamma_2, \Gamma_3, \Gamma_4\):
\begin{equation}
    (1 + \Tilde G^\dagger \Tilde G)^{-1} = (1 + G^\dagger G + G^\dagger \Delta G + \Delta G^\dagger G)^{-1}.
\end{equation}
Now we designate:
\begin{equation}\label{eq:a_def}
    A = 1 + G^\dagger G \, , \, \Delta A = G^\dagger \Delta G + \Delta G^\dagger G
\end{equation}
and rewrite 
\begin{equation}\label{eq:geom_ser1}
    (A + \Delta A)^{-1} = \left(\sum_{n=0}^\infty (-A^{-1}\Delta A)^n \right)A^{-1}.
\end{equation}
We can designate:
\begin{equation}
    S = \sum\limits_{n=1}^\infty (-A^{-1}\Delta A)^n
\end{equation}
and obtain:
\begin{equation}
\begin{aligned}\label{eq:geom_ser2}\nonumber
    (1 &+ \Tilde G^\dagger \Tilde G)^{-1} - (1 + G^\dagger G)^{-1} \\ &= \left(\sum_{n=1}^\infty (-A^{-1}\Delta A)^n \right)A^{-1} = SA^{-1}.
\end{aligned}
\end{equation}
Using this result we arrive at:
\begin{equation}\label{eq:gamma1_corrections}
    \Delta \Gamma_1 = S A^{-1},
\end{equation}
\begin{equation}\label{eq:gamma4_corrections}
\Delta \Gamma_4 = \Delta G A^{-1} G^\dagger + G A^{-1} \Delta G^\dagger + 
 G S A^{-1} G^\dagger,
\end{equation}
\begin{equation}\label{eq:gamma2_corrections}
    \Delta \Gamma_2 = \Delta G A^{-1} + G SA^{-1},
\end{equation}
\begin{equation}\label{eq:gamma3_corrections}
    \Delta \Gamma_3 = A^{-1} \Delta G^\dagger + SA^{-1} G^\dagger.
\end{equation}
In \cref{eq:geom_ser1,eq:geom_ser2} we used the inverse of \(A\), which exists since \(A\) is positive-definite. To ensure that the series \(S\) converges, we also used the fact that:
\begin{equation} \label{eq:normAdA}
\norm{A^{-1} \Delta A} < 1.    
\end{equation}
The letter follows from:
\begin{equation}\label{eq:s_ineq}
    \norm{A^{-1} \Delta A} \leq \norm{A^{-1}G^\dagger \Delta G} + \norm{A^{-1}\Delta G^\dagger G}.
\end{equation}
To obtain bounds on both terms, we need to bound each of the matrices in the equations. Their scaling behavior is presented in \cref{tab:e_g_asympt}. In this and the next table by \(\alpha, \beta \approx N\) we mean \(\abs{\alpha-N}, \abs{\beta-N} = \mathcal{O}(1)\).
\begin{table}[h]
\caption{Matrix asymptotics}
\label{tab:e_g_asympt}
\begin{center}
\begin{tabular}{ |L||L|L|L|L|L|L| } 
 \hline
 \phantom & V_{\alpha\beta} & E_\alpha - E_\beta & G_{\alpha,\alpha+1} & A^{-1}_{\alpha} & \Delta G_{\alpha\beta} & S_{\alpha\beta} \\ \hline \hline
 \alpha,\beta \approx N & N^{-3} & N^{-2} & N & N^{-2} & 1 & N^{-1}\\ \hline
 \text{other } \alpha, \beta & 0 & N^{-1} \text{ or } 1 & 1 & 1 & N^{-2} & N^{-2}\\
 \hline
\end{tabular}
\end{center}
\end{table}

The first column is just our choice of \(V\). The second column reflects the following: when \(\abs{\alpha-\beta} = \Theta(N)\), then \(\abs{E_\alpha - E_\beta} = \Theta(E_{max} - E_{min}) = \mathcal{O}(1)\); when  \(\abs{\alpha-\beta} \ll N\) and both indices are far from \(N\), then \(\abs{E_\alpha - E_\beta} \approx E_\alpha' (k) \diff{k} = \Theta(N^{-1})\). Finally, when \(\abs{\alpha-N}, \abs{\beta-N} = \mathcal{O}(1)\), then \(\abs{E_\alpha - E_\beta} = \mathcal{O}(E_\alpha'' (\pi) \diff{k^2}) = \mathcal{O}(N^{-2})\). The third column follows from \cref{eq:g_asymptot}. The fourth column is a consequence of the third column and definition of \(A\) \cref{eq:a_def}. For the fifth column we use \cref{eq:dG_matform,eq:a_perturb_general}:
\begin{equation}
    \Delta G_{\alpha\beta} = 2\frac{(V_1)_{\alpha\beta} G_{\beta\beta} - (V_1)_{\alpha\beta}^* G_{\alpha\alpha}}{E_\alpha - E_\beta} + \mathcal{O}(N^{-2}),
\end{equation}
and then apply the results from the first three columns. For the sixth column we use the columns 3-5 to both summands in \cref{eq:s_ineq}. This proves that \(S_{\alpha\beta}\) converges and gives the asymptotics of \(S_{\alpha\beta}\) from the table. 

Applying results from the table to  \cref{eq:gamma1_corrections,eq:gamma2_corrections,eq:gamma3_corrections,eq:gamma4_corrections} we derive the following for the fermionic correlators:
\begin{table}[h]
\caption{Matrix asymptotics}
\label{tab:correl_asympt}
\begin{center}
\begin{tabular}{ |L||L|L|L|L| } 
 \hline
 \phantom & (\Delta \Gamma_1)_{\alpha\beta} & (\Delta \Gamma_2)_{\alpha\beta} & (\Delta \Gamma_3)_{\alpha\beta} & (\Delta \Gamma_4)_{\alpha\beta}  \\ \hline \hline
 \alpha, \beta \approx N & N^{-1} & N^{-2} & N^{-2} & N^{-3} \\ \hline
 \text{other } \alpha, \beta & N^{-2} & N^{-2} & N^{-2} & N^{-2}\\
 \hline
\end{tabular}
\end{center}
\end{table}
\subsection{Fermions in the position basis}
Fermionic operators in the position basis and the energy basis are connected through a unitary transform. Therefore, for pairs of fermionic correlators we can write:
\begin{equation}\label{eq:correl_unitary}
    \Delta \mathcal{C} = U \Delta \Gamma U^\dagger,  
\end{equation}
where
\begin{equation}\label{eq:cs_st_structure}
\mathcal{C} = \begin{pmatrix}
    \mathcal{C}_1 & \mathcal{C}_2 \\
    \mathcal{C}_3 & \mathcal{C}_4
    \end{pmatrix} \; , \;
    \begin{cases}
    (\mathcal{C}_1)_{ij} = \langle c_i c_j^\dagger \rangle
\\
    (\mathcal{C}_2)_{ij} = \langle c_i c_j \rangle
\\
    (\mathcal{C}_3)_{ij} = \langle c_i^\dagger c_j^\dagger \rangle
\\
    (\mathcal{C}_4)_{ij} = \langle c_i^\dagger c_j \rangle.
\end{cases}
\end{equation}
Since \(\Delta \Gamma\) is a hermitian matrix, it is diagonalizable. Further, \(U\) is unitary, so we can bound the maximal entry of \(\Delta \mathcal{C}\) with the maximal eigenvalue of \((\Delta \Gamma)_{ij}\):
\begin{equation}
    \max\limits_{ij} \abs{\left(U \Delta \Gamma U^\dagger\right)_{ij}} \leq \max\limits_{\substack{\norm{v} = 1 \\ \norm{w} = 1}} \abs{v^\dagger (\Delta \Gamma) w} = \abs{\lambda_{max}}. 
\end{equation}
In its turn, the maximal eigenvalue can be bounded as:
\begin{equation}
    \lambda_{max} \leq \max_j \sum_i \abs{(\Delta \Gamma)_{ij}}.
\end{equation}
There are \(\mathcal{O}(N)\) entries in each row with the values scaling as \(\mathcal{O}(N^{-2})\) and \(\mathcal{O}(1)\) entries with the values \(\mathcal{O}(N^{-1})\) - see \cref{tab:correl_asympt}. Therefore, for any column of \(\Delta \Gamma\) the sum has asymptotics:
\begin{equation}
    \sum_i \abs{(\Delta \Gamma)_{ij}} = \mathcal{O}(N) \mathcal{O}(N^{-2}) + \mathcal{O}(1) \mathcal{O}(N^{-1}) = \mathcal{O}(N^{-1}).
\end{equation}

\section{Finite disorder}\label{sec:app:finitedisorder}

Our theory operates with perturbations of order \(\nu = \Theta(N^{-3})\). For such amplitudes we have theoretically established that new DQPTS already emerge but the change in fermionic and spin-spin correlators remains suppressed in the thermodynamic limit. In practice for \(N = 1000\) it means perturbation amplitude \(\nu \sim 10^{-4} - 10^{-5}\) - see \cref{fig:gij_numerical_disord}. Our theoretical bounds on change in correlators are not tight, because simulations show that correlators remain unchanged at much stronger perturbations, up to \(\nu \sim 10^{-2}\) - see \cref{fig:strong_dis}. At these amplitudes of random perturbations, correlators rapidly change their oscillation frequency to a new one, which does not correspond to any Fisher zero crossing we had before. 
\begin{figure}[h!]
    \centering
    \includegraphics[width=0.5\textwidth]{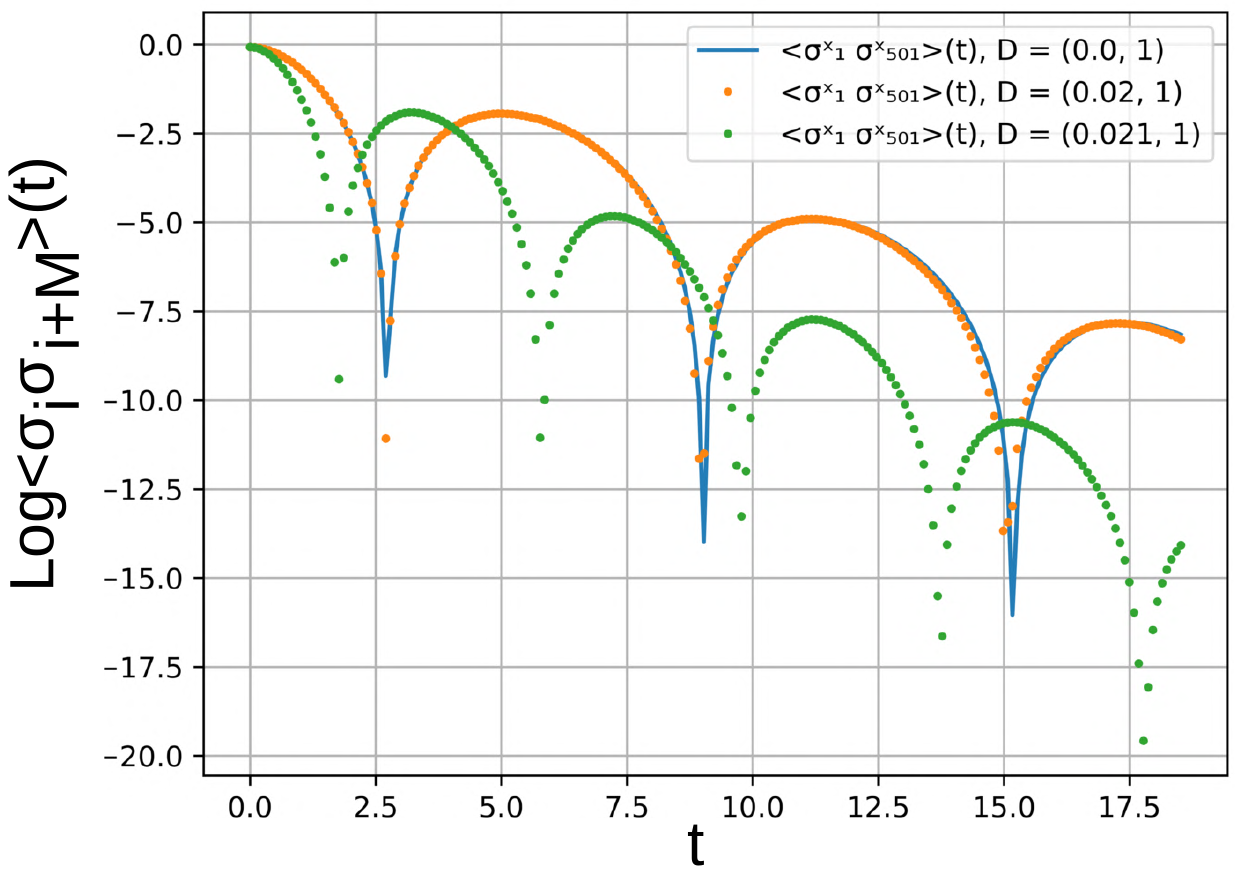}
    \caption{Spin-spin correlators, calculated for a quench with initial field \(h^0 = 0.5\), and final disordered field with a mean value \(h^1 = 1.5\) and a disorder amplitude \(D\), chain length \(N = 1000\). The correlators calculated for spins divided by \(d = N / 2 = 500\) spins. Presented cases are: for a homogeneous external magnetic field \(h^1\) (blue line), for a weakly disordered field that does not change correlators (orange dots), and for a field with larger disorder that changes correlators (green dots).}
    \label{fig:strong_dis}
\end{figure}
\begin{figure*}[!htb]
  \centering
    \includegraphics[width=0.8\textwidth]{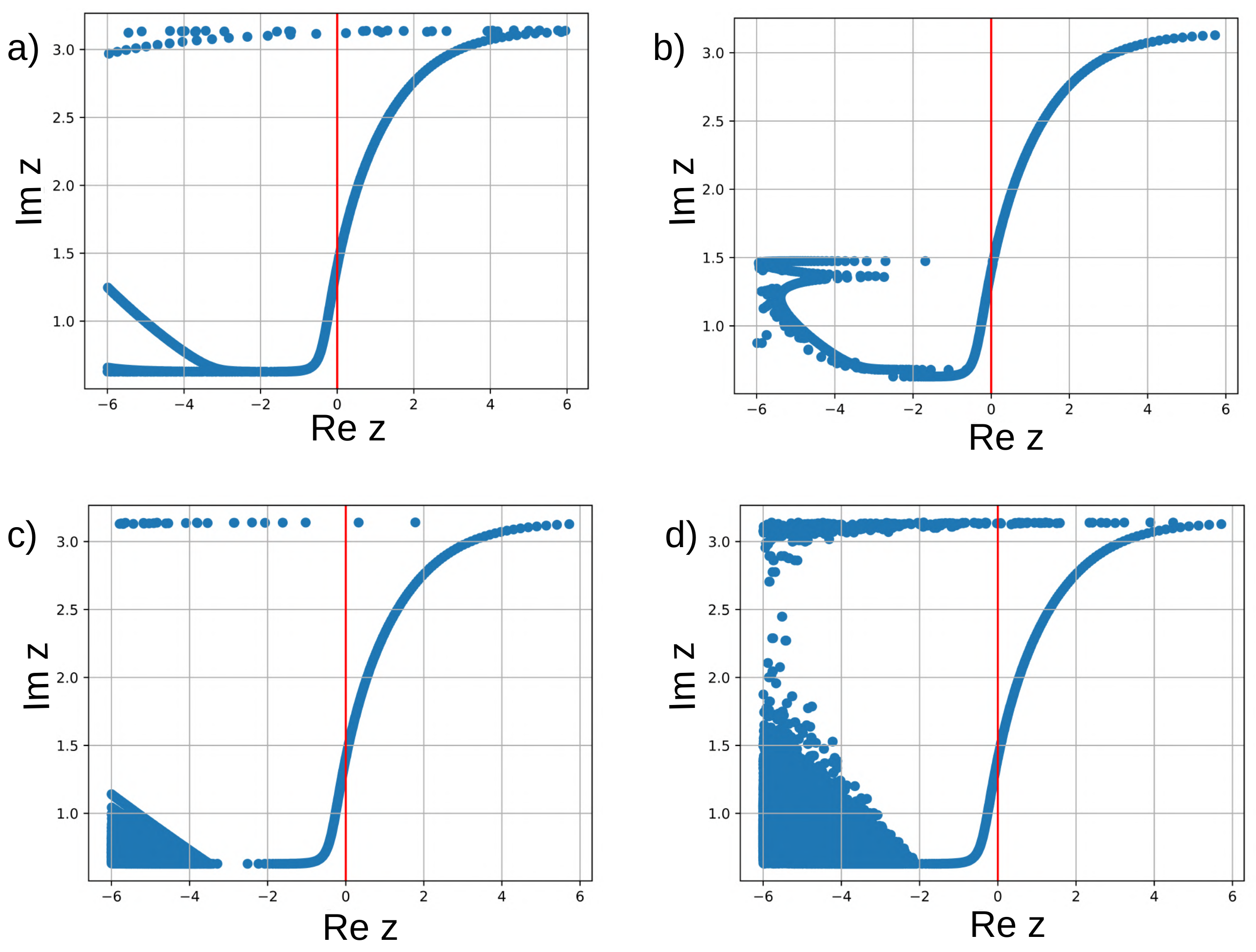}
    \caption{Fisher zeros for a quench with initial fields \(h^0 = 0\) and post-quench magnetic fields with the mean value \( h^1 = 1.5\) and with a weak harmonic or random perturbation \(\Delta h^1_n\), where \(n\) is a spin coordinate. Perturbation amplitude is \(D = 0.001\). \textbf{(a)} Perturbation with a small wave vector: \(\Delta h^1_n = D \cos(q n) \, , \, q = \frac{2 \pi}{N}\) \textbf{(b)} Perturbation with a large wave vector \(\Delta h^1_n = D \cos(q n) \, , \, q \approx \frac{\pi}{2}\) \textbf{(c)} Single site perturbation \(\Delta h^1_{n_0} = D \, , \, \Delta h_{n \neq n_0} = 0\) \textbf{(d)} Perturbation by a disordered potential \(\Delta h^1_n \in \mathcal{U}_{[-D,D]}\)}
    \label{fig:zeros_dif_freq}
\end{figure*}
\section{Perturbations of different wavelength}
In this section, we want to numerically study the effect of perturbations periodic in space with different wavelengths. It illustrates a statement made in the main text and, in particular, in \cref{subs:phys_pic}, that at a sufficiently low amplitude of the field modulation, it is only effective at changing Fisher zeros (and Loschmidt echo), if it has large wavelength components. In \cref{fig:zeros_dif_freq} we see, that when field modulation has Fourier components with small wave vectors, there are new Fisher zeros close to the imaginary axis. On the contrary, when modulation has only large wave vector components, no new Fisher zeros appear close to the imaginary axis. In the latter case, Fisher zeros far from the imaginary axis do not produce any new non-analyticities in the Loschmidt echo.
\begin{table*}[!htb]
\caption{At the intersection of a column titled with an operator A and a row titled with an operator B is  \([A, B]\)}
\label{tab:ferm_pair_comm}
\begin{center}
\begin{tabular}{ |c||c|c|c|c| } 
 \hline
 \phantom & \(\eta_\alpha \eta_\beta\) & \(\eta_\alpha^\dagger \eta_\beta\) & \(\eta_\alpha \eta_\beta^\dagger\) & \(\eta_\alpha^\dagger \eta_\beta^\dagger\) \\ \hline \hline
 \(\eta_\gamma \eta_\delta\)                 & 0 & \(\eta_\delta \eta_\beta \delta_{\alpha\gamma} - \eta_\gamma \eta_\beta \delta_{\alpha\delta}\) & \(\eta_\alpha \eta_\delta \delta_{\beta\gamma} - \eta_\alpha \eta_\gamma \delta_{\beta\delta}\) & \(\eta_\alpha^\dagger \eta_\delta \delta_{\beta\gamma} - \eta_\alpha^\dagger \eta_\gamma \delta_{\beta\delta} + \eta_\delta \eta_\beta^\dagger \delta_{\alpha\gamma} - \eta_\gamma \eta_\beta^\dagger \delta_{\alpha\delta}\) \\ \hline
 \(\eta_\gamma^\dagger \eta_\delta\)         & \(\eta_\alpha \eta_\delta \delta_{\beta\gamma} + \eta_\delta \eta_\beta \delta_{\alpha\gamma}\) & \(-\eta_\gamma^\dagger \eta_\beta \delta_{\alpha\delta} + \eta_\alpha^\dagger \eta_\delta \delta_{\beta\gamma}\) & \(\eta_\delta \eta_\beta^\dagger \delta_{\alpha\gamma} - \eta_\alpha \eta_\gamma^\dagger \delta_{\beta\delta}\) & \(- \eta_\gamma^\dagger \eta_\beta^\dagger \delta_{\alpha\delta} - \eta_\alpha^\dagger \eta_\gamma^\dagger \delta_{\beta\delta}\) \\ \hline
 \(\eta_\gamma \eta_\delta^\dagger\)             & \(- \eta_\gamma \eta_\beta \delta_{\alpha\delta} - \eta_\alpha \eta_\gamma \delta_{\beta\delta}\) & \(\eta_\delta^\dagger \eta_\beta \delta_{\alpha\gamma} - \eta_\alpha^\dagger \eta_\gamma \delta_{\beta\delta}\) & \(-\eta_\gamma \eta_\beta^\dagger \delta_{\alpha\delta} + \eta_\alpha \eta_\delta^\dagger \delta_{\beta\gamma}\) & \(\eta_\delta^\dagger \eta_\beta^\dagger \delta_{\alpha\gamma} + \eta_\alpha^\dagger \eta_\delta^\dagger \delta_{\beta\gamma}\) \\ \hline
 \(\eta_\gamma^\dagger \eta_\delta^\dagger\) & \(\eta_\alpha \eta_\delta^\dagger \delta_{\beta\gamma} - \eta_\alpha \eta_\gamma^\dagger \delta_{\beta\delta} + \eta_\delta^\dagger \eta_\beta \delta_{\alpha\gamma} - \eta_\gamma^\dagger \eta_\beta \delta_{\alpha\delta}\) & \(\eta_\alpha^\dagger \eta_\delta^\dagger \delta_{\beta\gamma} - \eta_\gamma^\dagger \eta_\alpha^\dagger \delta_{\beta\delta}\) & \(\eta_\delta^\dagger \eta_\beta^\dagger \delta_{\alpha\gamma} - \eta_\gamma^\dagger \eta_\beta^\dagger \delta_{\alpha\delta}\) & 0 \\
 \hline
\end{tabular}
\end{center}
\end{table*}
\section{Low energy part of BCS-matrix, numerical test}\label{sec:dgij_num_test}
In this subsection, we want to numerically test \cref{eq:gij_final_momentum}. We fix average post-quench field at \(h = 1.5\) and choose \(p_\alpha = 1 \, , \,  p_\beta = 2\), so that \(\Tilde{G}_{\alpha\beta}\) is maximal. Next, we put \(b_q = \frac{D}{\sqrt{N}}\), where \(\sqrt{N}\) is just a Fourier-normalization factor. Substituting these to \cref{eq:gij_final_momentum}, we obtain:
\begin{equation}\label{eq:gij_numerical}
    \log \left(\max\limits_{\alpha\beta}\abs{G_{\alpha\beta}}\right) = 2.5 \log N + \log D + C,
\end{equation}
where \(C \approx -6.729\). We conducted a series of numerical tests. First - with a fixed number of spins \(N = 1000\) and varied perturbation amplitude - see \cref{fig:gij_numerical_disord}. Second - with fixed perturbation amplitude \(D = 10^{-5}\) and varied number of spins - see \cref{fig:gij_numerical_disord}.

\end{document}